\documentclass[10pt]{article}

\usepackage[margin=1in]{geometry}
\usepackage{natbib}
\usepackage{parskip}

\usepackage{amssymb}
\usepackage{amsmath}
\usepackage{graphicx}
\usepackage{mathtools}
\usepackage{needspace} 
\usepackage{multicol}
\usepackage{listings}
\usepackage{amsthm}
\usepackage[shortlabels]{enumitem}
\usepackage[italicdiff]{physics}
\usepackage{cancel}
\usepackage[dvipsnames]{xcolor}
\usepackage{verbatim}
\usepackage{comment}
\usepackage{array}
\usepackage{tikz-cd}
\usepackage{import}
\usepackage{booktabs}

\usepackage{array}
\usepackage[colorlinks, linkcolor=blue, citecolor=ForestGreen]{hyperref}
\usepackage{forest}
\usepackage{float}
\usepackage{xspace}
\usepackage{aliascnt}
\usepackage{wrapfig}
\usepackage{thm-restate}
\usepackage{adjustbox}
\usepackage{bm}
\usepackage{bbm}
\usepackage{complexity}
\usepackage{makecell}
\usepackage{nicefrac}
\usetikzlibrary{calc, positioning, external, arrows.meta}
\usepackage{cleveref}
\usepackage{colortbl}
\usepackage{cellspace}

\allowdisplaybreaks

\let\E\relax

\DeclareMathOperator*{\E}{\mathbb E}

\let\cite\citep

\newcommand{\cA}{\mathcal{A}}

\newcommand{\cF}{\mathcal{F}}
\newcommand{\cG}{\mathcal{G}}

\newcommand{\cO}{\mathcal{O}}
\newcommand{\cT}{\mathcal{T}}
\newcommand{\cU}{\mathcal{U}}

\newcommand{\I}{\mathbb{I}}

\newcommand{\rev}{\textsf{Rev}\xspace}
\newcommand{\gft}{\textsf{GFT}\xspace}
\newcommand{\opt}{\textsf{Opt}}

\usepackage{algorithm}
\usepackage{algorithmicx}
\usepackage[noend]{algpseudocode}

\newfloat{protocol}{tbp}{lop}
\floatname{protocol}{Protocol}

\newtheorem{theorem}{Theorem}
\numberwithin{theorem}{section}
\newtheorem*{theorem*}{Theorem}

\newtheorem{lemma}[theorem]{Lemma}

\newtheorem{corollary}[theorem]{Corollary}

\theoremstyle{definition}
\newtheorem{definition}[theorem]{Definition}

\usepackage{caption}
\makeatletter
\renewcommand\thanks[1]{%
  \footnotemark%
  \protected@xdef\@thanks{\@thanks
    \protect\footnotetext[\the\c@footnote]{#1}}%
}
\makeatother

\title{Regret Minimization in Bilateral Trade With
Perturbed Markets}

\author{
	Anna Lunghi\\[-0.2em]
	\footnotesize\texttt{anna.lunghi@polimi.it} \\[-0.2em]
	\footnotesize Politecnico di Milano 
	\and
	Matteo Castiglioni\\[-0.2em]
	\footnotesize\texttt{matteo.castiglioni@polimi.it} \\[-0.2em]
	\footnotesize Politecnico di Milano 
	\and
	Alberto Marchesi \\[-0.2em]
	\footnotesize\texttt{alberto.marchesi@polimi.it} \\[-0.2em]
	\footnotesize Politecnico di Milano
}
\date{\today}

\begin{document}

\maketitle

\begin{abstract}

We address the problem of maximizing Gain from Trade (GFT) in repeated buyer-seller exchanges subject to global budget balance constraints. While this problem is well-understood in purely adversarial and stochastic settings, these environments exhibit a sharp dichotomy: adversarial environments allow for no-regret learning against the best fixed-price mechanism, whereas stochastic environments allow for no-regret learning against the best distribution over prices that is budget balanced in expectation. This gap is significant, as policies balanced in expectation can increase the GFT by a multiplicative factor of two. 
In this work, we bridge these extremes by studying perturbed markets, where an underlying stochastic distribution is subject to an adversarial corruption $C$. We design an algorithm that adaptively scales with the level of corruption, achieving an $\tilde{\mathcal{O}}(T^{3/4}) + \mathcal{O}(C\log(T))$ regret bound against the best budget-balanced distribution over prices. Simultaneously, our algorithm maintains the worst-case $\tilde{\mathcal{O}}(T^{3/4})$ regret bound relative to a per-round budget-balanced baseline, ensuring optimality even in fully adversarial environments.

\end{abstract}

\newpage

\section{Introduction}

The bilateral trade problem stands as a cornerstone of mechanism design, modeling the exchange of a good between a buyer and a seller through an intermediary \cite{Myerson1983Efficient}. While classic results primarily address the Bayesian setting \citep{BlumrosenD14,kang22fixed,Mcafee08,BlumrosenM16,brustle2017approximating,DengMSW21,fei2022improved}, starting from~\citep{cesa2021regret}, recent interest has shifted toward the repeated version of this problem, analyzed through the lens of online learning.

In the online version of the bilateral trade problem, a  buyer and a seller with valuations $b_t$ and $s_t$, respectively, arrive at each time step $t \in [T]$, with the learner implementing a mechanism that posts prices $(p_t, q_t)$ to the seller and the buyer, respectively. A trade occurs if both agents find the transaction profitable, \emph{i.e.}, $s_t\le p_t$ and $b_t\ge q_t$. The goal of the intermediary, \emph{i.e.}, the learner, is to maximize the cumulative Gain from Trade (GFT):
\[ \sum_{t=1}^T \gft_t(p_t,q_t) =\sum_{t =1}^T(b_t-s_t) \I(s_t\le p_t,b_t\ge q_t),\]
subject to some budget balance constraints related to not subsidizing the market.

Early works focused on mechanisms that satisfy budget balance---either strong or weak---on a per-round basis. Under these constraints, the learner is restricted to ensuring that its revenue, \emph{i.e.}, $(q_t-p_t) \I(s_t\le p_t,b_t\ge q_t)$  is either exactly zero or non-negative in every single interaction. We provide an extensive discussion on these works in \Cref{app:related}.
As notice by subsequent works, these constraints might be unnecessary restrictive. A more flexible paradigm was introduced by \cite{bernasconi2024glob}, which proposed the concept of Global Budget Balance (GBB). Under GBB, the mechanism is allowed to subsidize certain trades (incurring negative revenue) as long as these are offset by profits in other rounds. This cumulative approach ensures that, by the end of the time horizon, the protocol remains economically viable for the intermediary while significantly improving the GFT.


While the primary reason why \cite{bernasconi2024glob} introduces GBB constraints is driven by facilitating learning, it also shows that this can induce a large (in particular a $2$ multiplicative factor) increase in the achievable GFT.
Under traditional per-round constraints, the best the intermediary could do is to learn the best fixed price in hindsight. In contrast, GBB allows the learner to balance revenue across rounds, making the best fixed distribution over prices (budget-balanced in expectation) a more ambitious and natural benchmark. This distributional benchmark was first explored by \cite{bernasconi2024glob}, which established that such benchmark is unlearnable in adversarial environments and sublinear regret is impossible.

Motivated by this negative result, \cite{lunghi2026stronger} shows that in stochastic environments with sufficient regularity, sublinear regret of $\tilde{\cO}(T^{3/4})$ is attainable.
These contrasting results raise the following question: 
\begin{center}
    \emph{Can we design a ``best-of-both-world'' learning algorithm that effectively interpolates between the stochastic and adversarial regimes?}
\end{center}

We address this question by designing algorithms that operate optimally in ``perturbed markets'', which introduce a total corruption $C$ to an underlying stochastic environment. This framework allows us to bridge the gap between the two regimes: when $C=0$, we recover the purely stochastic setting, whereas $C=T$ captures the adversarial one. 

\subsection{Our Results and Techniques}

We design an algorithm that adapts to the nature of the environment, achieving the best-attainable guarantees across both stochastic, perturbed, and adversarial settings. We evaluate our learner against two benchmarks: (i) fixed-price benchmark $\opt^{\textnormal{F}}$, which is the cumulative GFT of the best fixed price; (ii) distributional benchmark $\opt^{\,\textnormal{D}}$, which is the cumulative GFT of the best GBB distribution.

Our algorithm ensures global budget balance, while providing the following regret guarantee:
\[ \sum_{t=1}^T \gft_t(p_t,q_t) \ge \max\left\{ \opt^{\,\textnormal{D}} - \mathcal{O}(C\log T), \, \opt^{\textnormal{F}} \right\} - \tilde{\mathcal{O}}(T^{3/4}). \]
This matches the optimal stochastic rates, remains robust to sublinear corruption $C$, and recovers the worst-case adversarial guarantees (see \Cref{table 1} for a schematic comparison).
\paragraph{Primal-Dual Methods and Bilateral Trade}
To the best of our knowledge, this work presents the first primal-dual framework applied to bilateral trade.
At its core lies the tension between maximizing the GFT and maintaining a non-negative cumulative revenue. We resolve this by alternating between two specialized procedures: (i) a primal-dual method that handles the constrained optimization of GFT and (ii) a revenue maximization algorithm, based on the one of \cite{bernasconi2024glob}, invoked to replenish the budget when violation arises.

\paragraph{Bypassing the Slater Condition}
In constrained online problems, primal-dual methods usually yield optimal rates only under Slater condition---which in our context would require the existence of a policy that generates strictly positive revenue at every round. This condition is clearly violated in bilateral trade. We demonstrate that this requirement can be replaced by a significantly weaker condition: a bounded ratio between the optimal GFT and the optimal revenue (which is at most $\mathcal{O}(\log T$)).
This allows us to relate the regret incurred during the revenue maximization phase directly to the collected revenue. We then link this revenue—which is in turn related to the violation of the primal-dual algorithm---to its negative regret. This ``surplus'' gain from trade generated during budget-violating rounds compensates for the regret suffered while replenishing the budget.

We also mention that one non-trivial component needed by our primal-dual method is the design of a novel primal regret minimizer that optimizes a linear combination of GFT and revenue. 

\paragraph{Why Perturbed Markets?}
We define regret relative to the corruption of a stochastic distribution. While alternative definitions of non-stationarity exist (\emph{e.g.}, total variation between consecutive rounds), these are too stringent for this problem. As shown by the lower bounds in \citep{bernasconi2024glob}, even a single change in the environment over the horizon can force linear regret against a distributional benchmark. This suggests that perturbed markets are the right tool to analyze ``nearly-stochastic'' environments.

\paragraph{Relation to Bandit With Knapsacks}
Our approach is conceptually related to ``best-of-both-world'' algorithms for bandits with knapsacks and more general constraints, where primal-dual methods are highly effective \cite{balseiro2020dual,fikioris2023approximately,bernasconi2024no,slivkins2023contextual, castiglioni2023online, castiglioni2022unifying}. However, this literature typically relies on the Slater condition to achieve optimal bounds. Without it, it is generally impossible to achieve a constant fraction of the optimal reward.
Our results can be viewed as a Slater-free result for bilateral trade, achieved by leveraging strict satisfiability in a global sense. In particular, we employ the $\log(T)$ upper bound on the ratio of reward to negative violation that holds globally.
The closest to our work is \citep{bernasconi2024bandits}, which uses a similar dual analysis to bound the primal violation as a function of its negative regret.

\begin{table}
\caption{Comparison with state of the art.}\label{table 1}
\centering
\renewcommand*{\arraystretch}{1.5}
 \begin{tabular}{l | c | c |c|c} 
  \toprule \label{tab:related}
    & \textbf{Baseline} & \textbf{Assumptions}  & \textbf{ Total Variation $C$} &\textbf{Regret}\\ 
   \midrule
  \citet{bernasconi2024glob} &$\opt^{\textnormal{F}}$ & None & $C\in [0,T]$ & $\tilde\cO(T^{3/4})$\\
  \hline
  \citet{bernasconi2024glob} &$\opt^{\,\textnormal{D}}$ & None & $C\in [0,T]$ & $\Omega(T)$\\
  \hline
  \citet{lunghi2026stronger} & $\opt^{\,\textnormal{D}}$ & $\sigma$-smooth & $C=0$ & $\tilde\cO(T^{3/4})$\\
  \midrule
  \cellcolor{gray!20}Our work & \cellcolor{gray!20} $\opt^{\,\textnormal{D}}$ & \cellcolor{gray!20}$\sigma$-smooth & \cellcolor{gray!20} $C\in [0,T]$ & \cellcolor{gray!20} $\tilde{\mathcal{O}}(T^{3/4}+C)$
  \\
  \hline
  \cellcolor{gray!20}Our work & \cellcolor{gray!20} $\opt^{\textnormal{F}}$ & \cellcolor{gray!20} None & \cellcolor{gray!20} $C\in [0,T]$ & \cellcolor{gray!20} $\tilde{\mathcal{O}}(T^{3/4})$
  \\
   \bottomrule
  \end{tabular}
\end{table}

\section{Preliminaries}

We study repeated bilateral trade through the lens of online learning. Over $T \in \mathbb{N}$ rounds, a sequence of seller-buyer pairs arrives. In each round $t$, the seller and buyer hold private valuations, denoted by $s_t \in [0,1]$ and $b_t \in [0,1]$, respectively. These valuations are realizations of a joint time-dependent probability distribution $\mathcal{L}_t$ with support in $[0,1]^2$, chosen by an oblivious adversary. 

\paragraph{Learning Protocol} At each round $t$, the learner posts two prices: $p_t \in [0, 1]$ to the seller and $q_t \in [0, 1]$ to the buyer.
A trade occurs if both parties accept---that is, if $s_t \le p_t$ and $b_t \ge q_t$. When a trade occurs, the seller receives $p_t$ and the buyer pays $q_t$; otherwise no transaction takes place.

To evaluate the performance of the mechanism, we consider two metrics: \emph{revenue}, which represents the profit or loss of the mechanism, and the \emph{Gain from Trade} (\gft), which serves as a proxy for the social welfare generated. Formally, for valuations $(s, b) \in [0,1]^2$ and prices $(p, q) \in [0,1]^2$, we define:
\begin{align*}
&\gft(p, q, s, b) := (b - s) \mathbb{I}(b \ge q, s \le p),\\
&\rev(p, q, s, b) := (q - p) \mathbb{I}(b \ge q, s \le p).
\end{align*}

For brevity, we define $\gft_t(p, q) := \gft(p, q, s_t, b_t)$ and $\rev_t(p, q) := \rev(p, q, s_t, b_t)$.

\paragraph{Feedback}
The learner does not observe the private valuations \( (s_t, b_t) \), either before or after posting prices; instead, only partial information is revealed through the feedback received after each round. We focus on the most restrictive setting, \textbf{one-bit feedback}. Under this model, the learner only observes whether the trade occurred, \emph{i.e.}, the realization of the indicator $\mathbb{I}(b_t \ge q_t, s_t \le p_t)$.

\paragraph{Budget Balance}

We consider the most relaxed notion of budget balance, which is called \emph{Global Budget Balance} (GBB). An algorithm is said to be GBB if its cumulative revenue is non-negative, \emph{i.e.}, $\sum_{t=1}^T \rev_t(p_t, q_t) \ge 0$. This constraint must hold almost surely regardless of the algorithm internal randomization and the realizations of the environment.

\paragraph{Baselines}

We consider two standard benchmarks in the literature. The first is the best fixed price, which is inherently budget-balanced on a per-round basis and represents the optimal performance achievable in adversarial environments \citep{bernasconi2024glob}. Formally, we define:
 \begin{equation}
    \opt^{\,\textnormal{F}} := \max_{p \in [0, 1]} \sum_{t=1}^T \gft_t(p, p) .\label{program: opt strong}
    \end{equation}

The second baseline is the best feasible distribution over price pairs that is budget-balanced in expectation. This serves as the target benchmark in stochastic settings. Formally, we define:
\begin{align}
        \opt^{\,\textnormal{D}} := \max_{\pi \in \Delta([0, 1]^2)} \quad & \sum_{t=1}^T \mathbb{E}_{(p, q) \sim \pi,(s,b)\sim \cL_t}[\gft(p, q,s,b)] \nonumber\\
        \text{s.t.} \quad & \sum_{t=1}^T \mathbb{E}_{(p, q) \sim \pi,(s,b)\sim \cL_t}[\rev_t(p, q,s,b)] \ge 0, 
        \label{program: opt glob}
    \end{align}
    where $\Delta([0,1]^2) $ represents the space of all the well-defined distributions over $[0,1]^2.$

\paragraph{Environment} We study \emph{perturbed markets} that interpolate between stochastic and adversarial settings. 
Let $\Delta([0,1]^2)$ be the set of distributions supported on $[0,1]^2$.
The environment is characterized by a probability distribution $\cP\in [0,1]^2$, representing the private valuations distribution in the unperturbed market.
Moreover, there is a perturbation parameter $C$, which is defined as the cumulative total variation distance between the sequence of distributions $\{\cL_t\}_{t =1}^T$ chosen by the adversary and the unperturbed stationary distribution $\cP$:
\[ C \coloneq 
\sum_{t=1}^T  \lVert \cL_t - \mathcal{P} \rVert_{\textnormal{TV}}. \]
Note that $C=0$ in a purely stochastic environment where all $\cL_t$ are identical, while $C$ can be as large as $T$ in the worst case (the adversarial setting).

When dealing with the distributional benchmark $\opt^{\,\textnormal{D}}$, as shown by \citet{lunghi2026stronger}, it is necessary to impose some additional regularity conditions on the environment. In particular, we assume 
the unperturbed distribution $\cP$ to be $\sigma$-smoothed:

\begin{definition}[\cite{haghtalab2024smoothed,cesa2023repeated}]\label{def:sigma}
    Let $X$ be a domain that supports a uniform distribution $\nu$. A measure $\mu$ on $X$ is said to be $\sigma$-smooth if for all measurable subsets $A\subseteq X$, we have $\mu(A)\le \frac{\nu(A)}{\sigma}$ 
\end{definition}

\section{Main Result and Proof Plan}

We design a best-of-both-world algorithm with optimal guarantees in both stochastic and adversarial settings. Furthermore, its performance degrades gracefully as a function of the perturbation.

\begin{theorem} 
\label{theo: main}
    There exists a GBB algorithm that, with probability $1-\cO(1/T)$,  guarantees:
    \textnormal{\[ \sum_{t=1}^T \gft_t(p_t,q_t)\ge \max\{ \opt^{\,\textnormal{D}}- O(C\log(T)), \opt^{\,\textnormal{F}}  \} - \tilde \cO(T^{3/4}) . \]}
\end{theorem}
Two considerations are in order. First, a linear dependency on $C$ is expected in this context, while we leave it as an open problem to determine whether the multiplicative $\log T$ factor is strictly necessary.
Second, our results recover the optimal regret rates of \cite{lunghi2026stronger} and \cite{bernasconi2024glob} for the stochastic and adversarial regimes, respectively. It is worth noting that the result in \cite{bernasconi2024glob} holds without the smoothness assumption. This is also true for our algorithm: the smoothed adversary assumption  is only required when competing against the stronger $\opt^{\,\textnormal{D}}$ baseline, where this assumption is required \cite{lunghi2026stronger}.

The remainder of this paper is organized as follows. \Cref{sec:discrete} shows how and when we can restrict to a finite grid of price pairs. \Cref{sec: main alg} presents the main algorithm, while \Cref{sec:analisisPD} presents our primal-dual algorithm. In \Cref{sec:primal}, we present the primal regret minimizer. Finally, in \Cref{sec:PET}, we combine the all the results to prove \Cref{theo: main}.

\section{From Continuous to Discrete Decision Space} \label{sec:discrete}

We begin by examining the conditions under which the continuous price space can be effectively approximated by a discrete grid. In bilateral trade, both the objective and the budget balance constraint are ill-behaved, \emph{e.g.}, they lack continuity. Hence, as already clear by the first works on the topic, see, \emph{e.g.}, \citep{cesa2024bilateral}, methods working on the continuum are ineffective, leaving grid-based approaches as the only feasible strategy.
However, without some form of regularity, grids are often ineffective.
This was already noted by \cite{cesa2021regret, cesa2023repeated} regarding the fixed-price benchmark when the learner is constrained by per-round budget balance prices. 
\cite{lunghi2026stronger} showed a similar result holds for the distributional benchmark, even when the learner can employ GBB policies.
In particular, both works show that it is impossible to achieve sublinear regret, even in a completely stochastic environment ($C=0$ in our setting). 

Both works circumvent these impossibility results by assuming $\sigma$-smoothness (see \Cref{def:sigma}).
Notice that this regularity assumption is required only when algorithm and baseline have comparable power. Indeed, \cite{bernasconi2024glob} shows that it is possible to achieve no regret with respect to the best fixed price using GBB policies, even without the $\sigma$-smoothness assumption. This is true also for our algorithm.  
With these considerations in mind, we follow the approach of \cite{lunghi2026stronger} and show that a uniform grid 
\[ \mathcal{G}_K \coloneq \left\{ \left( \frac{i}{K-1}, \frac{j}{K-1} \right) : i, j \in \{0, \dots, K-1\} \right\} \]
provides a high-quality approximation of the continuum when the distributions $\mathcal{L}_t$ are $\sigma$-smooth. 




Let $\Delta(\cG_K)$ be the space of all possible distributions of pair of prices on the grid $\cG_K$. 
Define:
    \begin{align}
        &\opt^{\,\textnormal{F}}_{K}= \begin{cases}
             \underset{\pi \in \Delta(\cG_K)}{\max} \,\,\underset{t=1}{\overset{T}{\sum}} \mathbb{E}_{(p, q) \sim \pi,(s,b)\sim\cL_t}[\gft(p, q,s,b)] \\
        \hspace{1cm}\text{s.t.} \;\quad \mathbb{E}_{(p, q) \sim \pi,(s,b)\sim \cL_t}[\rev_t(p, q,s,b)] \ge -\frac{1}{K} \quad \forall t\in [T] ,
        \end{cases}\label{program: opt zero k}
        \end{align}
        and
        \begin{align}
        &\opt^\textnormal{\,D}_K = \begin{cases}\underset{\pi \in \Delta(\cG_K)}{\max} \,\, \underset{t=1}{\overset{T}{\sum}} \mathbb{E}_{(p, q) \sim \pi,(s,b)\sim \cL_t}[\gft(p, q,s,b)] \\
        \hspace{1cm}\text{s.t.} \;\quad  \underset{t=1}{\overset{T}{\sum}} \mathbb{E}_{(p, q) \sim \pi,(s,b)\sim\cL_t}[\rev(p, q,s,b)] \ge 0. \end{cases}\label{program: opt k}
    \end{align}
 Intuitively, $\opt^\textnormal{\,D}_K$ requires prices to be budget balance in expectation over the rounds, similarly to GBB, while $\opt_{K}^\textnormal{\,F}$ requires the distribution to be almost budget balanced at each round. Then:
    
\begin{restatable}{lemma}{gridapprox}
\label{lemma: grid approx}
It holds:
     \textnormal{\(   R_{0,K}=\opt^{\,\textnormal{F}}-\opt^{\,\textnormal{F}}_{K}\le \tilde\cO\left(\frac{T}{K}\right).\)}
\end{restatable}

\begin{restatable}{lemma}{gridapproxglob}
\label{lemma: grid approx glob}
It holds:
   \textnormal{\(     R_K=\opt^{\,\textnormal{D}}-\opt^{\,\textnormal{D}}_K\le \tilde \cO\left(\frac{T}{K}+C\right),\)}
    where the notation $\cO\left(\cdot\right)$ contains the linear dependence on $\sigma$, which we assume constant. 
\end{restatable}

As promised, we need the $\sigma$-smoothness assumption only when dealing with the distributional benchmark. Similarly to \cite{bernasconi2024glob}, we remove the assumption for per-round budget balance prices thanks to GBB.
These results allow us  to restrict our focus to  a discrete uniform grid.

\section{A Primal-Dual Approach to Bilateral Trade}
\label{sec: main alg}

\begin{wrapfigure}{tr}{0.5\textwidth}
\vspace{-0.8cm}
\begin{minipage}{0.5\textwidth}
\begin{algorithm}[H]
  \caption{No-Regret Algorithm}\label{alg:main}
  \begin{algorithmic}[1]
  \State \textbf{Input: $T,\delta$} 
  \State $B_0\gets 0$,
  \For{$t= 1,\ldots, T$}
  \If{$B_{t-1}<1$}
  \State $(p_t,q_t)\gets$ \textsf{Rev-Max}($T,\delta$)
  \State Post $(p_t,q_t)$
   \Else 
   \State $(p_t,q_t)\gets \textsf{Primal-Dual}(T,\delta)$
   \State Post $(p_t,q_t)$
   \EndIf
   \State $B_{t}\gets B_{t-1}+\rev_t(p_t,q_t)$
   \EndFor
  \end{algorithmic}
\end{algorithm}
\end{minipage}
\end{wrapfigure}

In this section, we introduce our novel algorithm, whose pseudocode is detailed in \Cref{alg:main}. The algorithm alternates between two procedures: \textsf{Primal-Dual} and \textsf{Rev-Max}.
At the core of our approach is a high-level strategy inspired by \cite{bernasconi2024glob}, which involves alternating between a revenue maximization phase---designed to generate positive profit and recover lost budget as quickly as possible---and a policy that might violate the GBB constraint. Unlike \cite{bernasconi2024glob}, however, we cannot precompute the violations.
Hence, we must continuously alternate between these two policies.

\begin{wrapfigure}{r}{0.5\textwidth}
    \begin{minipage}{0.5\textwidth}
\begin{algorithm}[H]
\caption{\textsf{Primal-Dual}}\label{alg:primal-dual}
  \begin{algorithmic}[1]
  \State \textbf{Input: $T,\delta$}
  \State $\lambda_1\gets 0$
  \State $M\gets 16\log(T)$
      \State $\eta\gets 1/\sqrt{T}$
  \While{$t\le T$}
  \State Post $(p_t,q_t)\gets \textsf{Primal}$
  \State Feed $\I(p_t\le s,b_t\ge q)$ and $\lambda_t$ to \textsf{Primal}
  \State $\lambda_{t+1}\gets \Pi_{[0,M]}\left(\lambda_t - \eta \rev_{t}(p_t,q_t)\right)$
    \EndWhile
    \end{algorithmic}
\end{algorithm}
\end{minipage}
\end{wrapfigure}
It will be useful to define the set of rounds  $\cT_1\subseteq[T]:=\{1,\ldots, T\}$ in which \Cref{alg:main} chooses to follow procedure \textsf{Rev-Max} and $\cT_2\subseteq[T]$ as the set of rounds in which \Cref{alg:main} chooses to follow procedure \textsf{Primal-Dual}, such that $\cT_1 \cup \cT_2 = [T]$. 

The main component of our algorithm is a completely novel primal-dual algorithm that seeks to maximize the \gft while satisfying the budget constraint. The pseudocode of the algorithm is presented in \Cref{alg:primal-dual}, while its analysis is deferred to \Cref{sec:analisisPD}. To the best of our knowledge, this is the first primal-dual approach applied to bilateral trade.
Although we cannot show that the algorithm satisfies the constraint, we show that whenever the algorithm violates the budget, it incurs \emph{negative regret}. This negative regret provides a buffer, allowing the learner to ``lose'' rounds during the \textsf{Rev-Max} phase.

Instead, the procedure \textsf{Rev-Max} is specifically designed to recover the lost budget, while partially mitigating the loss in terms of GFT while doing so. 
We use the same procedure of  \cite{bernasconi2024glob} which consists in employing an algorithm of revenue maximization over a grid of prices of logarithmic nature. 
\cite{bernasconi2024glob} shows the following result. 
\begin{lemma}[Essentially \citep{bernasconi2024glob}] \label{lemma: profit max}
    There exists an algorithm \textnormal{\textsf{Rev-Max}} that, for any $K\in \mathbb{N}$, guarantees with probability at least $1-1/T$:
    \textnormal{\begin{equation*}
        \sum_{t\in \cT_1}\mathbb{E}_{(p,q)\sim \pi}[\gft_t(p,q)]\le 16 \log T\cdot \left(\sum_{t\in \cT_1}\rev_t(p_t,q_t)+1\right) + \tilde{\mathcal{O}}( T^{3/4})
    \end{equation*}}
     for all $\pi \in \Delta([0,1]^2)$ such that \textnormal{$\sum_{t=1}^T\mathbb{E}_{(p,q)\sim \pi,(s,b)\sim \cL_t}[\rev\,(p,q,s,b)]\ge 0$}.
\end{lemma}

Interestingly, this result is related to a relaxed Slater condition. Standard primal-dual methods \cite{bernasconi2024no, castiglioni2023online, castiglioni2022unifying,balseiro2020dual} 
 for settings with adversarial constraints require a positive and sufficiently large parameter
\begin{equation*}
    \rho = \max_{\pi \in \Delta([0,1]^2)} \min_{t\in[T]} \mathbb{E}_{(p,q)\sim \pi}[\rev_t(p,q)].
\end{equation*}
It is straightforward to observe that the Slater parameter $\rho=0$ is our setting. 
Despite that, our setting satisfies a weaker condition in which the relationship between the cumulative GFT and the accumulated budget has a global nature and does not require a strict per-round assumption.

\section{Regret Minimization via Primal-Dual Framework }\label{sec:analisisPD}

We frame the online bilateral trade problem as an online constrained optimization problem, which we Lagrangify obtaining the reward function:
\[L_t(\pi,\lambda) =\E_{(p,q)\sim \pi}[\gft_t(p,q)]+\lambda_t\E_{(p,q)\sim \pi}[\rev_t(p,q)].
\]
The algorithm works by instantiating two distinct regret minimizers:
(i) a primal procedure that optimizes the reward function $L_t$ without explicit budget balance constraints and (ii) a dual procedure that updates the multipliers to penalize constraint violations. 

While the primal procedure can be any algorithm that achieves no-regret with high probability, the dual multipliers must be selected by an algorithm that guarantees weakly-adaptive regret minimization.
Given two integers $a,b$, we define $[a,b]=\{a,a+1, \ldots, b-1,b\}$.
Specifically, the dual algorithm must work on the domain $[0, M]$, where $M = 16 \log T$, and ensure no-regret over all possible sub-intervals $[t_1, t_2] \subseteq [T]$. 
This can be easily achieved by using Online Gradient Descent, which provides the following guarantees (see, \emph{e.g.}, \citep{hazan2016introduction}):
\begin{align}\label{eq: dual guarantee}
\sum_{t=t_1}^{t_2}(\lambda_t-\lambda)\rev_t(p_t,q_t) \le R_T^D =\tilde{\cO}(\sqrt{T})\quad \forall \lambda \in [0,M], [t_1,t_2]\subseteq [T]. 
\end{align}
Instead, designing a primal regret minimizer is a not trivial task that we address in \Cref{sec:primal}. For now, we assume that with probability at least $1-\cO(1/T)$, the primal minimizer satisfies:
\begin{align}
    \sum_{t=1}^T &\left(\gft_t(p_t,q_t)+\lambda_t\rev_t(p_t,q_t)\right)\hspace{-0.05cm} -\hspace{-0.05cm}\sum_{t=1}^T \left(\E_{(p,q)\sim \pi}[\gft_t(p,q)+\lambda_t\rev_t(p,q)]\right)\le  R_T^P.\label{eq: primal reg}
\end{align}
Notice that we parametrize the results by $R_T^P$ and discuss the specific form of the regret bound $R_T^P$ in the following.

Reordering the terms of \Cref{eq: primal reg}, and considering that the primal regret minimizer is active on the subset of rounds $\cT_2$ we get that, given a $\pi \in \Delta(\cG_K)$, with probability at least $1-\cO(1/T)$:
\begin{align}
    \sum_{t\in \cT_2}&\mathbb{E}_{(p,q)\sim\pi}[\gft_t(p,q)]- \sum_{t\in \cT_2}\gft_t(p_t,q_t)\nonumber\\
       &\hspace{2cm} \le R_T^P + \sum_{t\in \cT_2}\lambda_t\rev_t(p_t,q_t) - \sum_{t\in \cT_2}\lambda_t\mathbb{E}_{(p,q)\sim\pi}[\rev_t(p,q)].\label{eq: primal 2}
\end{align}
In the following sections, we bound the term $\sum_{t\in \cT_2}\lambda_t\rev_t(p_t,q_t)$ exploiting the adaptivity of the dual regret minimizer and the term $-\sum_{t=1}^T\lambda_t\mathbb{E}_{(p,q)\sim\pi}[\rev_t(p,q)]$ exploiting the definition of $C$.

In the following sections, we will study separately the two regret terms on the right that depend on the revenue.

\subsection{Adaptivity of the Dual Update}
In this section, we bound the term $\sum_{t\in \cT_2}\lambda_t\rev_t(p_t,q_t)$ by exploitying the weak adaptivity of the dual regret minimizer.
This term relates the budget balance violations of the \textsf{Primal-Dual} procedure to its accumulated negative regret. Crucially, we require any budget violation to be ``sufficiently profitable'' to offset the regret incurred while replenishing the budget via the \textsf{Rev-Max} procedure.



This idea is conceptually similar to the approach employed by \cite{bernasconi2024bandits} for the related problem of bandits with knapsacks. 
We split $\cT_2$ into two intervals around the last activation round $\tau$ of \textsf{Rev-Max}. By construction, the budget at time $\tau$ is close to zero. Formally,
$B_\tau = \sum_{t\in \cT_1\cap [\tau]}\rev_t(p_t,q_t) + \sum_{t\in \cT_2\cap [\tau]}\rev_t(p_t,q_t) \in [1,2].$

Applying the dual regret guarantee in the interval $\cT_2\cap [\tau]$ with respect to $\lambda=M$, we can state:
\begin{align*}
\sum_{t\in \cT_2\cap [\tau]}\lambda_t\rev_t(p_t,q_t) \le M\left(\sum_{t\in \cT_2\cap[\tau]}\rev_t(p_t,q_t)\right)+R_T^D.
\end{align*}
For the remaining rounds in the second interval $\cT_2\cap[\tau+1,T]$, the term remains small because the dual algorithm guarantees no regret with respect to $\lambda=0$:
\begin{align*}
\sum_{t\in \cT_2\cap [\tau+1,T]}\lambda_t\rev_t(p_t,q_t) \le 0 + R_T^D.
\end{align*}
This leads to the following result, which employs the almost equivalence of the negative revenue in $\cT_2\cap[\tau]$ and the positive revenue in $\cT_1$ 
\begin{restatable}{lemma}{dualreg}\label{lemma: dual in regret}
\Cref{alg:main} guarantees:
    \textnormal{\begin{align*}
        \sum_{t\in \cT_2\cap [\tau]}\lambda_t\rev_t(p_t,q_t) \le M\left(-\sum_{t\in \cT_1}\rev_t(p_t,q_t)+2\right)+2R_T^D
    \end{align*}}
\end{restatable}
Notice that the simpler approach that does not employ weak adaptivity is doomed to fail. If, for example, the \textsf{Rev-Max} algorithm is forced to balance a temporary loss in the budget of the \textsf{Primal-Dual} algorithm—which the \textsf{Primal-Dual} would have recovered on its own later on—it generates a large sum $\sum_{t\in \mathcal{T}_1} \rev_t(p_t, q_t) + \sum_{t\in \mathcal{T}_2} \rev_t(p_t, q_t)$, while the individual term $\sum_{t\in \mathcal{T}_1} \rev_t(p_t, q_t)$ (intuitively related on the regret in rounds $\mathcal{T}_1$) also remains large.

\subsection{Dependence on $C$}
In this section, we discuss the term $-\sum_{t\in \cT_2} \lambda_t \mathbb{E}_{(p,q)\sim\pi}[\rev_t(p,q)]$ in \Cref{eq: primal 2}. This term is highly dependent on $C$, and the formal reason why our results for stochastic and adversarial settings differ. 
We split our analysis in two cases, depending on which benchmark $\pi$ we are comparing with.

If $\pi$ is a solution to Program~\ref{program: opt zero k}, then $\pi$ must satisfy the budget balance constraint at each round, so \[-\sum_{t\in \cT_2} \lambda_t \mathbb{E}_{(p,q)\sim\pi}[\rev_t(p,q)]\le -\sum_{t\in \cT_2}\lambda_t\left(-1/K\right)+\tilde{\cO}(\sqrt{T})\le M/K + \tilde{\cO}(\sqrt{T}),\] where $\tilde{\cO}(\sqrt{T})$ is a concentration factor that relates the realized revenue with the expected one.

Conversely, if $\pi$ is a solution to Program~\ref{program: opt k} the analysis becomes more complex. In particular, the fact that the Lagrangian multipliers can change at each round, makes the term $-\sum_{t\in\cT_2} \lambda_t \mathbb{E}_{(p,q)\sim\pi}[\rev_t(p,q)]$ highly sensible to the change in expected revenue at each episode, that now is bounded only globally over all rounds and not locally for each round. In particular, it can be showed that the term $-\sum_{t\in \cT_2} \lambda_t \mathbb{E}_{(p,q)\sim\pi}[\rev_t(p,q)]$ can be upper bounded using the total variation of the environment $C$.

\begin{restatable}{lemma}{lemmaC}\label{lemma: C in regret}
    Let $\pi$ be a solution to Program~\ref{program: opt zero k} , then with probability at least $1-1/T$:
    \textnormal{\begin{align*}
    -\sum_{t\in \cT_2}\lambda_t\mathbb{E}_{(p,q)\sim\pi}[\rev_t(p,q)]\le M\cdot \left(\frac{T}{K}\right)+\tilde{\cO}\left(\sqrt{T}\right).
\end{align*}}
Moreover, let $\pi$ be a solution to Program~\ref{program: opt k}, then with probability at least $1-1/T$:
\textnormal{\begin{align*}
     -\sum_{t\in \cT_2}\lambda_t\mathbb{E}_{(p,q)\sim\pi}[\rev_t(p,q)]\le 2 M\cdot C +  \tilde{\cO}\left(\sqrt{T}\right).
\end{align*}}
\end{restatable}





\section{Primal Algorithm and the Observability Challenge} \label{sec:primal}

In this section, we develop a primal regret minimizer designed to satisfy the requirements of \Cref{eq: primal reg} with an optimal regret bound $R^P_T$.
The algorithm works over a set of price pairs $\cA$---specifically, we will apply the result for the grid $\cG_K$---and provides guarantees that scale linearly with the maximum value attained by the Lagrangian multipliers $M=16\log(T)$.

\begin{algorithm}
  \caption{\textsf{Primal Algorithm}}\label{alg:primal}
  \begin{algorithmic}[1]
  \State \textbf{Input:} $\cA,\eta^P,T,\gamma,\alpha,M$
  \State $w_1(p,q)\gets 1 \quad \forall (p,q)\in \cA$
  \While{$t <T$}
  \State Sample $(\hat{p},\hat{q})\sim \hat{\pi}_t$
  \State Sample $H_t\in \{0,1,2\}$, such that $\mathbb{P}(H_t=0)=1-\alpha, \mathbb{P}(H_t=1)=\mathbb{P}(H_t=1)=\alpha/2$
\State Post prices 
  \begin{align*}
      (p_t,q_t)=\begin{cases}
          (\hat{p}_t,\hat{q}_t) \quad &\text{ if } H_t=0\\
          (U_t\sim \cU([0,1]),\hat{q_t}) & \text{ if } H_t=1\\
          (\hat{p}_t,V_t\sim \mathcal{U}([0,1])) & \text{ if } H_t=2
      \end{cases}
  \end{align*}
  \State Receive feedback $\I(s_t\le p_t, b_t\ge q_t)$ and lagrangian multiplier $\lambda_t$
  \State Build loss feedback $\tilde{\ell}_t(p,q)$ for all $(p,q)\in \cA$ as in \Cref{eq: reward} 
  \State $w_{t+1}(p,q)\gets w_t(p,q)\exp\left(-\eta^P \tilde \ell_ t(p,q)\right)\quad \forall (p,q)\in \cA$
  \State $\hat{\pi}_{t+1}(p,q)\gets \frac{w_{t+1}(p,q)}{\sum_{(p',q')\in \cA}w_{t+1}(p',q')}\quad \forall (p,q)\in \cA$
  \EndWhile
  \end{algorithmic}
\end{algorithm}

A fundamental challenge in designing learning algorithms for bilateral trade is the observability challenge. Unlike standard learning settings, feedback here is significantly more limited: even basic bandit feedback is missing. 
To overcome this limited observability, we leverage the decomposition introduced by \cite{lunghi2026stronger}. We partition the GFT into three components: a ``bandit'' term (the revenue) and two ``global'' terms that can be estimated independently by probing each side of the market.

\begin{lemma}[Lemma 8.1, \cite{lunghi2026stronger}]
\label{lemma: L R decomposition}
Let $U,V \sim \cU([0,1])$ be independent, where $\cU([0,1])$ is the uniform distribution over $[0,1]$. The expected GFT of prices $(p,q)$ can be decomposed as:
\textnormal{\begin{align*}
\gft_t(p,q) = \underbrace{L_t(p,q)}_{\text{\emph{Seller Term}}} + \underbrace{R_t(p,q)}_{\text{\emph{Buyer Term}}} + \underbrace{\rev_t(p,q)}_{\text{\emph{Bandit Term}}},
\end{align*}}
where $L_t(p,q) \coloneq \mathbb{E}_{U}[\mathbb{I}(s_t \le U \le p, q \le b_t)]$ and $R_t(p,q) \coloneq \mathbb{E}_{V}[\mathbb{I}(s_t \le p, q \le V \le b_t)]$.
\end{lemma}


Let $\mathcal{A}_s \coloneq \{p : (p, q) \in \mathcal{A}\}$ and $\mathcal{A}_b \coloneq \{q : (p, q) \in \mathcal{A}\}$ be the set of unique seller and buyer prices, respectively, induced by $\mathcal{A}$.
We will show that $L_t(p,q)$ can be estimated simultaneously for all $p \in \mathcal{A}_s$ whenever $q$ is fixed, and $R_t(p,q)$ can be estimated for all $q \in \mathcal{A}_b$ whenever $p$ is fixed. This, however, require a ``full exploration`` strategy, which we invoke with probability $\alpha$.
By sharing feedback across actions that share a common price, we significantly reduce the estimator variance and hence the number of rounds dedicated to exploration.

Formally, our algorithm employs an \textsf{EXP3.IX}-style update with loss estimator $\tilde{\ell}_t(p,q)$ that employs importance sampling combined with an implicit exploration bias \cite{kocak2014efficient,neu2015explore}.
\begin{equation}\label{eq: reward}
    \tilde{\ell}_t(p,q) = \begin{cases}\displaystyle
        \frac{1-\I(s_t\le U_t\le p, b_t\ge q)}{\frac{\alpha}{2}\sum_{p': (p',q)\in \mathcal{A}}\hat{\pi}_t(p',q)+\gamma}&\textnormal{if }q=\hat q_t \text{ and }H_t=1  \vspace{0.2cm}\\
    \displaystyle\frac{1-\I(s_t\le p, b_t\ge V_t\ge q)}{\frac{\alpha}{2}\sum_{q': (p,q')\in \mathcal{A}}\hat{\pi}_t(p,q')+\gamma}& \textnormal{if }p= \hat{p}_t \text{ and }H_t=2 \vspace{0.2cm}\\
      \displaystyle\frac{(1+\lambda_t)\left(1-(q_t-p_t)\I(s_t\le \hat{p}_t, b_t\ge \hat q_t)\right)}{(1-\alpha)\hat{\pi}_t(p,q)+\gamma}& \text{if }(p,q)= (\hat p_t,\hat q_t)\text{ and } H_t=0 \vspace{0.2cm}\\
    0 &\textnormal{otherwise.}
    \end{cases}
\end{equation}

For this choice of $\tilde{\ell}_t$ and learning algorithm we can guarantee the following:

\begin{restatable}{theorem}{RegretPrimal}
\label{theo: RegretPrimal}
For appropriate choices of $\eta^P$ and $\gamma$, and letting $M \ge \max_{t \in [\tau]} |\lambda_t|$, \Cref{alg:primal} ensures with probability at least $1 - \cO(1/T)$:
    \textnormal{\begin{align*}
        \sum_{t=1}^{\tau}&\sum_{(p,q)\in \cA}\pi(p,q)\left(\gft_t(p,q)+\lambda_t\rev_t(p,q)\right)-\sum_{t=1}^{\tau}\sum_{(p,q)\in \cA}\left(\gft_t(p_t,q_t)+\lambda_t\rev_t(p_t,q_t)\right)\\
        &\hspace{5.5cm}\le \tilde{\cO}\left(M\left(\sqrt{|\cA|T}+\frac{|\cA_s|+|\cA_b|}{\sqrt{|\cA|}}\frac{\sqrt{T}}{\alpha}+\alpha \cdot T\right)\right).
    \end{align*}}
\end{restatable}




\section{Putting Everything Together}\label{sec:PET}

Finally, we can  prove \Cref{theo: main}. While the full proof is deferred to the appendix, we provide a sketch here.
Set $K=T^{-1/4},\alpha=T^{-1/4}$, $\cA=\cG_K$, $M=16\log(T)$.
We consider two cases:
\paragraph{Distributional Benchmark With $\sigma$-Smooth Adversary}
Let $\pi$ be the solution of Program~\ref{program: opt glob}, by \Cref{lemma: profit max} with probability at least $1-\cO(1/T)$
\begin{align*}
    \sum_{t\in \cT_1}& \mathbb{E}_{(p,q)\sim\pi}[\gft_t(p,q)]  -\sum_{t\in \cT_1}\gft_t(p_t,q_t)\le M\sum_{t\in \cT_1}\rev_t(p_t,q_t) +  \tilde{\cO}(T^{3/4}).
\end{align*}
By \Cref{lemma: C in regret}, \Cref{lemma: dual in regret}, and \Cref{lemma: grid approx glob} with probability at least $1-\cO(1/T)$:
\begin{align*}
    \sum_{t\in \cT_2}&\left(\mathbb{E}_{(p,q)\sim\pi}[\gft_t(p,q)]-\gft_t(p_t,q_t)\right)\le \tilde{\cO}(T^{3/4})-M\hspace{-0.1cm}\sum_{t\in \cT_1}\rev_t(p_t,q_t)- \cO(\log(T)C).
\end{align*}
Hence, with probability at least $1-\cO(1/T)$:
\begin{align*}
    \opt^{\,\textnormal{D}} - \sum_{t\in [T]}\left(\gft_t(p_t,q_t)\right)\le \tilde{\cO}\left(T^{3/4}\right) + \left(\log(T)C\right).
\end{align*}
\paragraph{Fixed-Price Benchmark}
If we use a solution of Program~\ref{program: opt strong} as a benchmark, we can replace \Cref{lemma: grid approx glob} with \Cref{lemma: grid approx}, and the first part (instead of the second as before) of \Cref{program: opt k} to get:
\begin{align*}
    \opt^{\textnormal{F}} - \sum_{t\in [T]}\gft_t(p_t,q_t)\le \tilde{\cO}\left(T^{3/4}\right) .
\end{align*}

\bibliographystyle{plainnat}
\bibliography{references.bib}

\begin{thebibliography}{40}
\providecommand{\natexlab}[1]{#1}
\providecommand{\url}[1]{\texttt{#1}}
\expandafter\ifx\csname urlstyle\endcsname\relax
  \providecommand{\doi}[1]{doi: #1}\else
  \providecommand{\doi}{doi: \begingroup \urlstyle{rm}\Url}\fi

\bibitem[Azar et~al.(2024)Azar, Fiat, and Fusco]{azar2022alpha}
Yossi Azar, Amos Fiat, and Federico Fusco.
\newblock An $\alpha$-regret analysis of adversarial bilateral trade.
\newblock \emph{Artificial Intelligence}, 337:\penalty0 104231, 2024.

\bibitem[Babaioff et~al.(2024)Babaioff, Frey, and Nisan]{babaioff2024learning}
Moshe Babaioff, Amitai Frey, and Noam Nisan.
\newblock Learning to maximize gains from trade in small markets.
\newblock In \emph{Proceedings of the 25th ACM Conference on Economics and
  Computation}, pages 195--195, 2024.

\bibitem[Bachoc et~al.(2024)Bachoc, Cesa-Bianchi, Cesari, and
  Colomboni]{bachoc2024fair}
Fran{\c{c}}ois Bachoc, Nicol{\`o} Cesa-Bianchi, Tom Cesari, and Roberto
  Colomboni.
\newblock Fair online bilateral trade.
\newblock \emph{Advances in Neural Information Processing Systems},
  37:\penalty0 37241--37263, 2024.

\bibitem[Bachoc et~al.(2025{\natexlab{a}})Bachoc, Cesari, and
  Colomboni]{bachoc2025a}
Fran{\c{c}}ois Bachoc, Tommaso Cesari, and Roberto Colomboni.
\newblock A contextual online learning theory of brokerage.
\newblock In \emph{Forty-second International Conference on Machine Learning},
  2025{\natexlab{a}}.

\bibitem[Bachoc et~al.(2025{\natexlab{b}})Bachoc, Cesari, and
  Colomboni]{bachoc2025b}
Fran{\c{c}}ois Bachoc, Tommaso Cesari, and Roberto Colomboni.
\newblock A tight regret analysis of non-parametric repeated contextual
  brokerage.
\newblock In \emph{The 28th International Conference on Artificial Intelligence
  and Statistics}, 2025{\natexlab{b}}.
\newblock URL \url{https://openreview.net/forum?id=U2RgUAySB6}.

\bibitem[Balseiro et~al.(2020)Balseiro, Lu, and Mirrokni]{balseiro2020dual}
Santiago Balseiro, Haihao Lu, and Vahab Mirrokni.
\newblock Dual mirror descent for online allocation problems.
\newblock In \emph{International Conference on Machine Learning}, pages
  613--628. PMLR, 2020.

\bibitem[Bernasconi et~al.(2024{\natexlab{a}})Bernasconi, Castiglioni, Celli,
  Fusco, et~al.]{bernasconi2024bandits}
M~Bernasconi, M~Castiglioni, A~Celli, F~Fusco, et~al.
\newblock Bandits with replenishable knapsacks: the best of both worlds.
\newblock In \emph{12th International Conference on Learning Representations,
  ICLR 2024}. International Conference on Learning Representations, ICLR,
  2024{\natexlab{a}}.

\bibitem[Bernasconi et~al.(2024{\natexlab{b}})Bernasconi, Castiglioni, Celli,
  and Fusco]{bernasconi2024glob}
Martino Bernasconi, Matteo Castiglioni, Andrea Celli, and Federico Fusco.
\newblock No-regret learning in bilateral trade via global budget balance.
\newblock In \emph{Proceedings of the 56th Annual ACM Symposium on Theory of
  Computing}, pages 247--258, 2024{\natexlab{b}}.

\bibitem[Bernasconi et~al.(2025)Bernasconi, Castiglioni, and
  Celli]{bernasconi2024no}
Martino Bernasconi, Matteo Castiglioni, and Andrea Celli.
\newblock No-regret is not enough! bandits with general constraints through
  adaptive regret minimization.
\newblock In \emph{International Conference on Machine Learning}, pages
  3877--3898. PMLR, 2025.

\bibitem[Blumrosen and Dobzinski(2014)]{BlumrosenD14}
Liad Blumrosen and Shahar Dobzinski.
\newblock Reallocation mechanisms.
\newblock In \emph{{EC}}, page 617. {ACM}, 2014.
\newblock \doi{10.1145/2600057.2602843}.

\bibitem[Blumrosen and Mizrahi(2016)]{BlumrosenM16}
Liad Blumrosen and Yehonatan Mizrahi.
\newblock Approximating gains-from-trade in bilateral trading.
\newblock In \emph{{WINE}}, volume 10123 of \emph{Lecture Notes in Computer
  Science}, pages 400--413. Springer, 2016.
\newblock \doi{10.1007/978-3-662-54110-4_28}.

\bibitem[Boli\'{c} et~al.(2024)Boli\'{c}, Cesari, and
  Colomboni]{bolic2024brokerage}
Natasa Boli\'{c}, Tommaso Cesari, and Roberto Colomboni.
\newblock An online learning theory of brokerage.
\newblock In \emph{Proceedings of the 23rd International Conference on
  Autonomous Agents and Multiagent Systems ({AAMAS})}, page 216–224, 2024.
\newblock ISBN 9798400704864.

\bibitem[Brustle et~al.(2017)Brustle, Cai, Wu, and
  Zhao]{brustle2017approximating}
Johannes Brustle, Yang Cai, Fa~Wu, and Mingfei Zhao.
\newblock Approximating gains from trade in two-sided markets via simple
  mechanisms.
\newblock In \emph{Proceedings of the 2017 ACM Conference on Economics and
  Computation}, pages 589--590, 2017.

\bibitem[Castiglioni et~al.(2022)Castiglioni, Celli, Marchesi, Romano, and
  Gatti]{castiglioni2022unifying}
Matteo Castiglioni, Andrea Celli, Alberto Marchesi, Giulia Romano, and Nicola
  Gatti.
\newblock A unifying framework for online optimization with long-term
  constraints.
\newblock \emph{Advances in Neural Information Processing Systems},
  35:\penalty0 33589--33602, 2022.

\bibitem[Castiglioni et~al.(2023)Castiglioni, Celli, and
  Kroer]{castiglioni2023online}
Matteo Castiglioni, Andrea Celli, and Christian Kroer.
\newblock Online learning under budget and roi constraints via weak adaptivity.
\newblock \emph{arXiv preprint arXiv:2302.01203}, 2023.

\bibitem[Castiglioni et~al.(2026)Castiglioni, Lunghi, and
  Marchesi]{castiglioni2026sample}
Matteo Castiglioni, Anna Lunghi, and Alberto Marchesi.
\newblock The sample complexity of uniform approximation for multi-dimensional
  cdfs and fixed-price mechanisms.
\newblock \emph{arXiv preprint arXiv:2602.10868}, 2026.

\bibitem[Cesa-Bianchi et~al.(2021)Cesa-Bianchi, Cesari, Colomboni, Fusco, and
  Leonardi]{cesa2021regret}
Nicol{\`o} Cesa-Bianchi, Tommaso~R Cesari, Roberto Colomboni, Federico Fusco,
  and Stefano Leonardi.
\newblock A regret analysis of bilateral trade.
\newblock In \emph{Proceedings of the 22nd ACM Conference on Economics and
  Computation}, pages 289--309, 2021.

\bibitem[Cesa-Bianchi et~al.(2023)Cesa-Bianchi, Cesari, Colomboni, Fusco, and
  Leonardi]{cesa2023repeated}
Nicol{\`o} Cesa-Bianchi, Tommaso~R Cesari, Roberto Colomboni, Federico Fusco,
  and Stefano Leonardi.
\newblock Repeated bilateral trade against a smoothed adversary.
\newblock In \emph{The Thirty Sixth Annual Conference on Learning Theory},
  pages 1095--1130. PMLR, 2023.

\bibitem[Cesa-Bianchi et~al.(2024{\natexlab{a}})Cesa-Bianchi, Cesari,
  Colomboni, Fusco, and Leonardi]{cesa2024bilateral}
Nicol{\`o} Cesa-Bianchi, Tommaso Cesari, Roberto Colomboni, Federico Fusco, and
  Stefano Leonardi.
\newblock Bilateral trade: A regret minimization perspective.
\newblock \emph{Mathematics of Operations Research}, 49\penalty0 (1):\penalty0
  171--203, 2024{\natexlab{a}}.

\bibitem[Cesa-Bianchi et~al.(2024{\natexlab{b}})Cesa-Bianchi, Cesari,
  Colomboni, Fusco, and Leonardi]{cesa2024regret}
Nicol{{\`o}} Cesa-Bianchi, Tommaso Cesari, Roberto Colomboni, Federico Fusco,
  and Stefano Leonardi.
\newblock Regret analysis of bilateral trade with a smoothed adversary.
\newblock \emph{Journal of Machine Learning Research}, 25\penalty0
  (234):\penalty0 1--36, 2024{\natexlab{b}}.

\bibitem[Cesari and Colomboni(2025)]{cesari2025an}
Tommaso Cesari and Roberto Colomboni.
\newblock An online learning theory of trading-volume maximization.
\newblock In \emph{The Thirteenth International Conference on Learning
  Representations}, 2025.

\bibitem[Chen et~al.(2025)Chen, Jin, Lu, and Zhang]{chen2025tight}
Houshuang Chen, Yaonan Jin, Pinyan Lu, and Chihao Zhang.
\newblock Tight regret bounds for fixed-price bilateral trade.
\newblock \emph{arXiv preprint arXiv:2504.04349}, 2025.

\bibitem[Coccia et~al.()Coccia, Bernasconi, and Celli]{coccianonparametric}
Emanuele Coccia, Martino Bernasconi, and Andrea Celli.
\newblock Nonparametric contextual online bilateral trade.
\newblock In \emph{The Fourteenth International Conference on Learning
  Representations}.

\bibitem[Cosson et~al.(2026)Cosson, Fusco, Gupta, Leonardi, Leme, and
  Russo]{cosson2026contextual}
Romain Cosson, Federico Fusco, Anupam Gupta, Stefano Leonardi, Renato~Paes
  Leme, and Matteo Russo.
\newblock Contextual online bilateral trade.
\newblock \emph{arXiv preprint arXiv:2602.12903}, 2026.

\bibitem[Deng et~al.(2022)Deng, Mao, Sivan, and Wang]{DengMSW21}
Yuan Deng, Jieming Mao, Balasubramanian Sivan, and Kangning Wang.
\newblock Approximately efficient bilateral trade.
\newblock In \emph{{STOC}}, pages 718--721. {ACM}, 2022.
\newblock \doi{10.1145/3519935.3520054}.

\bibitem[Di~Gregorio et~al.(2025)Di~Gregorio, Dütting, Fusco, and
  Schwiegelshohn]{digregorio2025}
Simone Di~Gregorio, Paul Dütting, Federico Fusco, and Chris Schwiegelshohn.
\newblock Nearly tight regret bounds for profit maximization in bilateral
  trade.
\newblock In \emph{2025 IEEE 66th Annual Symposium on Foundations of Computer
  Science (FOCS)}, pages 1570--1594, 2025.
\newblock \doi{10.1109/FOCS63196.2025.00083}.

\bibitem[Fei(2022)]{fei2022improved}
Yumou Fei.
\newblock Improved approximation to first-best gains-from-trade.
\newblock In \emph{International Conference on Web and Internet Economics},
  pages 204--218. Springer, 2022.
\newblock \doi{10.1007/978-3-031-22832-2_12}.

\bibitem[Fikioris and Tardos(2023)]{fikioris2023approximately}
Giannis Fikioris and {\'E}va Tardos.
\newblock Approximately stationary bandits with knapsacks.
\newblock In \emph{The Thirty Sixth Annual Conference on Learning Theory},
  pages 3758--3782. PMLR, 2023.

\bibitem[Gaucher et~al.(2025)Gaucher, Bernasconi, Castiglioni, Celli, and
  Perchet]{gaucher2025featurebased}
Solenne Gaucher, Martino Bernasconi, Matteo Castiglioni, Andrea Celli, and
  Vianney Perchet.
\newblock Feature-based online bilateral trade.
\newblock In \emph{The Thirteenth International Conference on Learning
  Representations}, 2025.
\newblock URL \url{https://openreview.net/forum?id=xnF2U0ro7b}.

\bibitem[Haghtalab et~al.(2024)Haghtalab, Roughgarden, and
  Shetty]{haghtalab2024smoothed}
Nika Haghtalab, Tim Roughgarden, and Abhishek Shetty.
\newblock Smoothed analysis with adaptive adversaries.
\newblock \emph{Journal of the ACM}, 71\penalty0 (3):\penalty0 1--34, 2024.

\bibitem[Hazan(2016)]{hazan2016introduction}
Elad Hazan.
\newblock Introduction to online convex optimization.
\newblock \emph{Foundations and Trends in Optimization}, 2\penalty0
  (3-4):\penalty0 157--325, 2016.

\bibitem[Kang et~al.(2022)Kang, Pernice, and Vondr{\'{a}}k]{kang22fixed}
Zi~Yang Kang, Francisco Pernice, and Jan Vondr{\'{a}}k.
\newblock Fixed-price approximations in bilateral trade.
\newblock In \emph{{SODA}}, pages 2964--2985. {SIAM}, 2022.
\newblock \doi{10.1137/1.9781611977073.115}.

\bibitem[Koc{\'a}k et~al.(2014)Koc{\'a}k, Neu, Valko, and
  Munos]{kocak2014efficient}
Tom{\'a}{\v{s}} Koc{\'a}k, Gergely Neu, Michal Valko, and R{\'e}mi Munos.
\newblock Efficient learning by implicit exploration in bandit problems with
  side observations.
\newblock \emph{Advances in Neural Information Processing Systems}, 27, 2014.

\bibitem[Lunghi et~al.()Lunghi, Castiglioni, and Marchesi]{lunghiBetter}
Anna Lunghi, Matteo Castiglioni, and Alberto Marchesi.
\newblock \emph{Better Regret Rates in Bilateral Trade via Sublinear Budget
  Violation}, pages 6494--6536.
\newblock \doi{10.1137/1.9781611978971.233}.
\newblock URL \url{https://epubs.siam.org/doi/abs/10.1137/1.9781611978971.233}.

\bibitem[Lunghi et~al.(2025)Lunghi, Castiglioni, and
  Marchesi]{lunghi2025online}
Anna Lunghi, Matteo Castiglioni, and Alberto Marchesi.
\newblock Online two-sided markets: Many buyers enhance learning.
\newblock \emph{arXiv preprint arXiv:2503.01529}, 2025.

\bibitem[Lunghi et~al.(2026)Lunghi, Piccinato, Castiglioni, and
  Marchesi]{lunghi2026stronger}
Anna Lunghi, Mattia Piccinato, Matteo Castiglioni, and Alberto Marchesi.
\newblock A stronger benchmark for online bilateral trade: From fixed prices to
  distributions.
\newblock \emph{arXiv preprint arXiv:2602.05681}, 2026.

\bibitem[McAfee(2008)]{Mcafee08}
R~Preston McAfee.
\newblock The gains from trade under fixed price mechanisms.
\newblock \emph{Applied economics research bulletin}, 1\penalty0 (1):\penalty0
  1--10, 2008.

\bibitem[Myerson and Satterthwaite(1983)]{Myerson1983Efficient}
Roger~B Myerson and Mark~A Satterthwaite.
\newblock Efficient mechanisms for bilateral trading.
\newblock \emph{Journal of economic theory}, 29\penalty0 (2):\penalty0
  265--281, 1983.

\bibitem[Neu(2015)]{neu2015explore}
Gergely Neu.
\newblock Explore no more: Improved high-probability regret bounds for
  non-stochastic bandits.
\newblock \emph{Advances in Neural Information Processing Systems}, 28, 2015.

\bibitem[Slivkins et~al.(2023)Slivkins, Sankararaman, and
  Foster]{slivkins2023contextual}
Aleksandrs Slivkins, Karthik~Abinav Sankararaman, and Dylan~J Foster.
\newblock Contextual bandits with packing and covering constraints: A modular
  lagrangian approach via regression.
\newblock In \emph{The Thirty Sixth Annual Conference on Learning Theory},
  pages 4633--4656. PMLR, 2023.

\end{thebibliography}

\newpage


\appendix

\section{Additional Related Works} \label{app:related}

The study of GFT maximization in online bilateral trade was initiated by \cite{cesa2024bilateral}, who show that under one-bit feedback the problem is in general unlearnable. However, they showed that Strong Budget Balanced (SBB) mechanisms are learnable when the seller and buyer valuation distributions are independent and admit bounded density. Subsequent works shifted toward adversarial environments. \cite{azar2022alpha} showed that no-2-regret is achievable, while \cite{cesa2024regret} introduced the smoothed-adversary assumption to online bilateral trade, proving that no-regret can be achieved under this assumption.

\cite{bernasconi2024glob} and \cite{chen2025tight} showed how GBB can be used to bypass the smoothed-adversary assumption.
Despite this, these works mainly focus on competing with the best SBB fixed price. The most relevant results to our work is a lower bound from \cite{bernasconi2024glob} showing that no-regret with respect to the best GBB distribution is unachievable in adversarial settings.
More recently, \cite{lunghiBetter} analyze the trade-off between the regret and the GBB constraint violation. More closely related to our work is
\cite{lunghi2026stronger}, which showed that it is possible to achieve no regret with respect to the distributional benchmark in stochastic setting under a smoothness assumption.

Other related lines of work focus on revenue maximization \cite{castiglioni2026sample, digregorio2025}. Finally, we mention that many works extended bilateral trade more complex problems such as contextual settings~\citep{gaucher2025featurebased,cosson2026contextual,coccianonparametric}, fairer objectives~\citep{bachoc2024fair},  multiple buyers~\citep{babaioff2024learning,lunghi2025online},  and settings with symmetric roles~\citep{bolic2024brokerage,cesari2025an,bachoc2025a,bachoc2025b}.

\section{Omitted Proofs From \Cref{sec:discrete}}

To prove that the approximation on the grid $\cG_K$ is a good approximation of the optimum on the continuum, we will proceed in two steps:
 first, we show that there exists a distribution on the grid that is slightly unfeasible but achieves the optimum GFT up to an additive factor $\cO(\frac{T}{K})$, then we show that we can interpolate between this policy and a high-profitable policy to build a policy that is feasible and maintains the same order distance from the optimum.
 
This procedure is highly similar to the one proposed by \cite{lunghi2026stronger}, slightly adapted to the non-stationarity of the environment, for what concern the best distribution.
\begin{lemma}\label{lemma: aux discr}
    Given the sequence of  distributions $\{\cL_t\}_{t=1}^T$,  all $\gamma\in \Delta([0,1]^2)$ that satisfies 
    \[\sum_{t=1}^T \E_{(p,q)\sim {\gamma},(s,b)\sim \cL_t}[\rev(p,q,s,b)]\ge 0,\] it exists a $\hat{\gamma}\in \Delta(\cG_K)$ that satisfies
    \begin{align*}
        \sum_{t=1}^T \E_{(p,q)\sim \hat{\gamma},(s,b)\sim \cL_t}[\gft(p,q,s,b)]\ge \sum_{t=1}^T \E_{(p,q)\sim \gamma,(s,b)\sim \cL_t}[\gft(p,q,s,b)] -\frac{2\sigma T}{K}-2C,
    \end{align*}
    and 
    \begin{align*}
        \sum_{t=1}^T \E_{(p,q)\sim \hat{\gamma},(s,b)\sim \cL_t}[\rev(p,q,s,b)]\ge  -\frac{2\sigma T}{K}-2C.
    \end{align*}
\end{lemma}
\begin{proof}
Let $\cP$ be a $\sigma$-smooth distribution .
    Let $(p,q)\in[0,1]^2$, and let $(\Pi_K^s(p),\Pi_K^b(q))$ denote the nearest point in the grid $\mathcal{G}_K$ such that
$\Pi_K^s(p)\ge p$ and $\Pi_K^b(q)\le q$.

We decompose the gain-from-trade as follows:
\begin{align*}
    \E_{(s,b)\sim \cP}[GFT(p,q,s,b)]
    &= \mathbb{E}_{(s,b)\sim \cP}\big[(b-s)\mathbb{I}(s\le p,\, b\ge q)\big] \\
    &= \mathbb{E}_{(s,b)\sim \cP}\big[(b-s)\mathbb{I}(s\le \Pi_K^s(p),\, b\ge \Pi_K^b(q))\big] \\
    &\quad - \mathbb{E}_{(s,b)\sim \cP}\big[(b-s)\mathbb{I}(p< s\le \Pi_K^s(p),\, b< q)\big] \\
    &\quad - \mathbb{E}_{(s,b)\sim \cP}\big[(b-s)\mathbb{I}(s\le p,\, \Pi_K^b(q)\le b< q)\big] \\
    &\quad + \mathbb{E}_{(s,b)\sim \cP}\big[(b-s)\mathbb{I}(p\le s< \Pi_K^s(p),\, \Pi_K^b(q)< b\le q)\big].
\end{align*}

Using $|b-s|\le 1$ and $\sigma$-smoothness, we obtain
\begin{align*}
    \E_{(s,b)\sim \cP}[GFT(p,q,s,b)]
    &\le \mathbb{E}\big[(b-s)\mathbb{I}(s\le \Pi_K^s(p),\, b\ge \Pi_K^b(q))\big] \\
    &\quad + \mathbb{P}(p\le s< \Pi_K^s(p),\, q\le b)
    + \mathbb{P}(s< p,\, \Pi_K^b(q)< b\le q) \\
    &= \E_{(s,b)\sim \cP}[GFT(\Pi_K^s(p),\Pi_K^b(q),s,b)] + \frac{2\sigma}{K}.
\end{align*}

Given any feasible distribution $\gamma\in\Delta([0,1]^2)$, define the discretized distribution
$\hat{\gamma}_K$ by
\begin{align*}
    \hat{\gamma}_K(p,q)= \mathbb{P}_{(x,y)\sim \gamma}\left((x,y)\in \bigg(p-\frac{1}{K},p\bigg]\times \bigg[q,q+\frac{1}{K}\bigg)\right) \quad \forall(p,q)\in \mathcal{G}_K
\end{align*}

Then,
\begin{align*}
    \mathbb{E}_{(p,q)\sim\gamma,(s,b)\sim \cP}[GFT(p,q,s,b)]
    &\le \mathbb{E}_{(p,q)\sim\gamma,(s,b)\sim \cP}\left[GFT(\Pi_K^s(p),\Pi_K^b(q),s,b)+\frac{2\sigma}{K}\right]\\
    &\le \mathbb{E}_{(p,q)\sim\hat{\gamma}_K,(s,b)\sim \cP}[GFT(p,q,s,b)]
    + \frac{2\sigma}{K}.
\end{align*}

A similar argument yields
\begin{align*}
    \E_{(s,b)\sim \cP}[\rev(p,q,s,b)]
    &\le \E_{(s,b)\sim \cP}[\rev(\Pi_K^s(p),\Pi_K^b(q),s,b)]
    + \mathbb{P}_{(s,b)\sim \cP}(p\le s< \Pi_K^s(p),\, q\le b) \\
    &\quad + \mathbb{P}_{(s,b)\sim \cP}(s< p,\, \Pi_K^b(q)< b\le q) \\
    &\le \E_{(s,b)\sim \cP}[\rev(\Pi_K^s(p),\Pi_K^b(q),s,b)] + \frac{2\sigma}{K},
\end{align*}
and therefore
\begin{align*}
    \mathbb{E}_{(p,q)\sim\gamma,{(s,b)\sim \cP}}[\rev(p,q,s,b)]
    \le \mathbb{E}_{(p,q)\sim\hat{\gamma}_K,{(s,b)\sim \cP}}[\rev(p,q,s,b)]
    + \frac{2\sigma}{K}.
\end{align*}
Then, if we consider $\overline{\cL}$ to be the average over all rounds $[T]$ of the distributions $\{\cL_t\}_t$, so that 
\[\lVert \overline{\cL}-\cP\lVert_{TV}\le \frac{C}{T}\]
for all $\gamma$
\begin{align*}
\sum_{t=1}^T\mathbb{E}_{(p,q)\sim\gamma,{(s,b)\sim \cL_t}}[\rev(p,q,s,b)]
    \le \sum_{t=1}^T\mathbb{E}_{(p,q)\sim\hat{\gamma}_K,{(s,b)\sim \cL_t}}[\rev(p,q,s,b)]
    + \frac{2\sigma}{K} + 2C.
\end{align*}
and we can use the feasibility of $\gamma$ to conclude the proof:
\[
\sum_{t=1}^T\mathbb{E}_{(p,q)\sim\hat{\gamma}_K,(s,b)\sim \cL_t}[\rev(p,q,s,b)]
\ge -\frac{2\sigma T}{K}-2C.
\]
\end{proof}

\gridapprox*
\begin{proof}
    The statement is directly mutated from \cite{bernasconi2024glob}. Specifically, it exists  $(p^\dagger,q^\dagger)\in \{(p,q)\in \cG_K: p-q=\frac{1}{K}\}$, such that , by the same argument in \cite{bernasconi2024glob}, and by definition of $\opt^{\textnormal{F}}_K$
    \begin{align*}
         \sum_{t=1}^T\gft_t(p^\dagger,q^\dagger)\ge \max_{p\in [0,1]}\sum_{t=1}^T \gft_t(p,p) -\frac{T}{K}\ge \opt^{\textnormal{F}}_K-\frac{T}{K}.
    \end{align*}
\end{proof}
\gridapproxglob*
\begin{proof}[Proof Sketch]
In general the proof is identical to Lemma 5.2 of \cite{lunghi2026stronger},
where in this case we distinguish two scenarios. 
\begin{itemize}
    \item It exist $\phi\in \Delta(\cG_K)$ such that \begin{align*}
        \sum_{t=1}^T\E_{(p,q)\sim \phi,(s,b)\sim \cL_t}[\rev(p,q,s,b)]\ge \frac{2\sigma T}{K}+2C.
    \end{align*}
    In this case we can build a well define mixture distribution between $\hat{\gamma}_k$ defined by \Cref{lemma: aux discr} and $\phi$, that respect the statement.
    \item It does \emph{not} exist $\phi\in \Delta(\cG_K)$ such that \begin{align*}
        \sum_{t=1}^T\E_{(p,q)\sim \phi,(s,b)\sim \cL_t}[\rev(p,q,s,b)]\ge \frac{2\sigma T}{K}+2C.
    \end{align*}
    In this case, using the fact that maximum revenue achievable and the optimum achievable GFT are linked up to a $\cO(\log(T)$ multiplicative factor,  it is obvious that any feasible distribution cannot be worse than the optimum of $\tilde{\cO}(\frac{T}{K}+C)$.
\end{itemize}

\end{proof}

\section{Omitted proofs from \Cref{sec:analisisPD}}



\dualreg*
\begin{proof}
    The proof of this lemma uses the alternating strategy of \cref{alg:main} between procedures and the weakly adaptive nature of the dual algorithm. 
    
    In particular, let $\tau$ be the last round in which \textsf{Rev-Max} is activated, i.e. $\tau=\max\{t: t\in \cT_1\}$. 
Since $\tau\in \cT_1$ this means that $B_{\tau}\le B_{\tau-1}+1\le 2$, and that for all $t\ge \tau$ we have that $B_t\ge 1$. 
These considerations implies the following:
\begin{align*}
    \sum_{t\in \cT_2\cap [\tau]}\rev_t(p_t,q_t)+\sum_{t\in \cT_1\cap [\tau]}\rev_t(p_t,q_t)\in [1,2]
\end{align*}
and
\begin{align*}
    \sum_{t\in \cT_1\cap [\tau]}\rev_t(p_t,q_t)= \sum_{t\in \cT_1}\rev_t(p_t,q_t).
\end{align*}

Therefore, we can can use the dual regret guarantees in the following way:
\begin{align*}
    \sum_{t\in \cT_2\cap [\tau]}\lambda_t\rev_t(p_t,q_t)\le M\left(\sum_{t\in \cT_2\cap[\tau]}\rev_t(p_t,q_t)\right)+R_T^D
\end{align*}
and
\begin{align*}
    \sum_{t\in \cT_2\cap [\tau+1,T]}\lambda_t\rev_t(p_t,q_t)\le 0\cdot\left(\sum_{t\in \cT_2\cap[\tau+1,T]}\rev_t(p_t,q_t)\right)+R_T^D= R_T^D.
\end{align*}
To conclude, 
\begin{align*}
    \sum_{t\in \cT_2}\lambda_t\rev_t(p_t,q_t)&\le \sum_{t\in \cT_2\cap [\tau]}\lambda_t\rev_t(p_t,q_t)+\sum_{t\in \cT_2\cap [\tau+1,T]}\lambda_t\rev_t(p_t,q_t)\\
    & \le M\sum_{t\in \cT_2\cap[\tau]}\rev_t(p_t,q_t)+2R_T^D \\
    & \le -M\sum_{t\in \cT_1}\rev_t(p_t,q_t) + 2M +2 R_T^D.
\end{align*}
\end{proof}

\lemmaC*

\begin{proof}
We split the proof in two cases.

\paragraph{Case 1: $\pi$ Solution of $\opt_K^F$}

If $\pi$ is a solution $\opt_K^F$, then by definition $\E_{(p,q)\sim \pi,(s,b)\sim \cL_t}[\rev(p,q,s,b)]\ge \frac{-1}{K}$, for all $t\in [T]$.
Therefore, by applying Azuma-Hoeffding inequality, with probability at least $1-1/T$
\begin{align*}
    \sum_{t\in \cT_2}&\lambda_t\mathbb{E}_{(p,q)\sim\pi}[\rev_t(p,q)]= \sum_{t\in \cT_2}\lambda_t\left(\mathbb{E}_{(p,q)\sim\pi}[\rev_t(p,q)]-\mathbb{E}_{(p,q)\sim\pi,(s,b)\sim \cL_t}[\rev_t(p,q,s,b)]\right)\\
    & + \sum_{t\in \cT_2}\lambda_t\left(\mathbb{E}_{(p,q)\sim\pi,(s,b)\sim \cL_t}[\rev_t(p,q,s,b)]\right)\\
    & \ge -M \cdot\tilde{\cO}\left(\sqrt{T}\right) -M |\cT_2| \frac{1}{K}.
\end{align*}

\paragraph{Case 2: $\pi$ Solution of $\opt_K^D$}
    \begin{align*}
        \sum_{t\in \cT_2}&\lambda_t\mathbb{E}_{(p,q)\sim\pi}[\rev_t(p,q)]\\
        &= \sum_{t\in \cT_2}\lambda_t\left(\mathbb{E}_{(p,q)\sim\pi}[\rev_t(p,q)]-\mathbb{E}_{(p,q)\sim\pi}[\overline{\rev}(p,q)]+\mathbb{E}_{(p,q)\sim\pi}[\overline{\rev}(p,q)]\right)\\
        & = \sum_{t\in \cT_2}\lambda_t\left(\mathbb{E}_{(p,q)\sim\pi}[\rev_t(p,q)]-\mathbb{E}_{(p,q)\sim\pi}[\overline{\rev}(p,q)]\right)+\sum_{t\in \cT_2}\lambda_t\left(\mathbb{E}_{(p,q)\sim\pi}[\overline{\rev}(p,q)]\right)\\
        & = \sum_{t\in \cT_2}\lambda_t\left(\mathbb{E}_{(p,q)\sim\pi}[\rev_t(p,q)]-\mathbb{E}_{(p,q)\sim\pi}[\E_{(s,b)\sim \cL_t}[[\rev(p,q,s,b)]]\right)\\
        &+ \sum_{t\in \cT_2}\lambda_t\left(\mathbb{E}_{(p,q)\sim\pi}[\E_{(s,b)\sim \cL_t}[[\rev(p,q,s,b)]]-\mathbb{E}_{(p,q)\sim\pi}[\overline{\rev}(p,q)]\right)\\
        & +\sum_{t\in \cT_2}\lambda_t\left(\mathbb{E}_{(p,q)\sim\pi}[\overline{\rev}(p,q)]\right)
    \end{align*}
Additionally, by definition of $C$
    \begin{align*}
        &\sum_{t\in \cT_2}\bigg|\mathbb{E}_{(p,q)\sim\pi}\left[\E_{(s,b)\sim \cL_t}[\rev(p,q,s,b)]-\E_{(s,b)\sim \cP}[\rev(p,q,s,b)]\right]\bigg|\le C
    \end{align*}
and by definition of feasibility of $\pi$
\begin{align*}
    \sum_{t\in [T]}\left(\mathbb{E}_{(p,q)\sim\pi}\E_{(s,b)\sim \cL_t}[\rev(p,q,s,b)]\right)=0.
\end{align*}
We can use $C$ also to bound the following quantity
\begin{align*}
    \bigg|\sum_{t\in [T]}\left(\mathbb{E}_{(p,q)\sim\pi}\E_{(s,b)\sim \cL_t}[\rev(p,q,s,b)]\right)-\sum_{t\in [T]}\left(\mathbb{E}_{(p,q)\sim\pi}\E_{(s,b)\sim \cP}[\rev(p,q,s,b)]\right)\bigg|\le C,
\end{align*}
which means that 
\begin{align*}
    \left(\mathbb{E}_{(p,q)\sim\pi}\E_{(s,b)\sim \cP}[\rev(p,q,s,b)]\right)\ge \frac{-C}{T}.
\end{align*}

Let $\lambda_t$ conditioned to the filtration up to $t-1$, be independent from the revenue at time $t$ conditioned to the filtration up to $t-1$, i.e. $\lambda_t$ is independent from $\cL_t$ and its realization, which exactly what happens in the primal-dual scheme.
    \begin{align*}
        \E&[\lambda_t\left(\mathbb{E}_{(p,q)\sim\pi}[\rev_t(p,q)]-\mathbb{E}_{(p,q)\sim\pi}[\E_{(s,b)\sim \cL_t}[[\rev(p,q,s,b)]]\right)|\cF_{t-1}]\\
        &= \E[\lambda_t|\cF_{t-1}]\E[\left(\mathbb{E}_{(p,q)\sim\pi}[\rev_t(p,q)]-\mathbb{E}_{(p,q)\sim\pi}[\E_{(s,b)\sim \cL_t}[[\rev(p,q,s,b)]]\right)|\cF_{t-1}].
    \end{align*}

    Then, by applying Azuma-Hoeffding with probability at least $1-1/T$
    \begin{align*}
        \sum_{t\in \cT_2}&\lambda_t\mathbb{E}_{(p,q)\sim\pi}[\rev_t(p,q)]\ge -2 M\cdot C -M\cdot \tilde{\cO}\left(\sqrt{T}\right).
    \end{align*}

\end{proof}






\section{Omitted Proof from \Cref{sec:primal}}

In this section, we provide the proof of \Cref{theo: RegretPrimal}.
Our analysis adapts the implicit-exploration framework of \cite{neu2015explore}, and share some similarities with the literature of Multi armed bandits with side information in the definition of the estimators (e.g. \cite{kocak2014efficient}). In this section we provide some auxiliary results related to the primal algorithm loss estimators and their second moment. Finally, we will use these result to prove \Cref{theo: RegretPrimal}.

\subsection{Primal Loss Estimators and Their Guarantees}
As a first step, we analyze the primal loss vector and its estimators.

We employ as true primal loss 
$\ell_t(p,q)= \underbrace{(1-L_t(p,q))}_{L_t^-(p,q)}+\underbrace{(1-R_t(p,q))}_{R_t^-(p,q)}+(\lambda_t+1)\underbrace{(1-\rev_t(p,q))}_{\rev_t^-(p,q)}$. This choice has the benefit of guaranteeing that each of its component is greater or equal to zero and the introduced bias with respect to the Lagrangian problem does not affect the general regret.

For each component, we can define two estimators, one unbiased and based on important sampling principle, indicated by a $\hat \cdot$, and one identical but for the addition on the denominator of the implicit exploration parameter $\gamma$, indicated by $\tilde \cdot$.

More formally 

\begin{equation*}
        \tilde L_t^-(p,q)=\frac{1-\I(s_t\le U_t\le p, b_t\ge q)}{\frac{\alpha}{2}\sum_{p': (p',q)\in \mathcal{A}}\hat{\pi}_t(p',q)+\gamma}\I(q=\hat q_t,H_t=1)
    \end{equation*}
    \begin{equation*}
        \hat L_t^-(p,q)=\frac{1-\I(s_t\le U_t\le p, b_t\ge q)}{\frac{\alpha}{2}\sum_{p': (p',q)\in \mathcal{A}}\hat{\pi}_t(p',q)}\I(q=\hat q_t,H_t=1)
    \end{equation*}
    \begin{equation*}
        \tilde R_t^-(p,q)=\frac{1-\I(s_t\le p, b_t\ge V_t\ge q)}{\frac{\alpha}{2}\sum_{q': (p,q')\in \mathcal{A}}\hat{\pi}_t(p,q')+\gamma}\I(p= p_t,H_t=2)
    \end{equation*}
    \begin{equation*}
        \hat R_t^-(p,q)=\frac{1-\I(s_t\le p, b_t\ge V_t\ge q)}{\frac{\alpha}{2}\sum_{q': (p,q')\in \mathcal{A}}\hat{\pi}_t(p,q')}\I(p= p_t,H_t=2)
    \end{equation*}
    \begin{equation*}
        \widetilde{\rev}^-_t(p,q)= \frac{1-(q_t-p_t)\I(s_t\le \hat{p}_t, b_t\ge \hat q_t)}{(1-\alpha)\hat{\pi}_t(p,q)+\gamma}\I((p,q)= (\hat p_t,\hat q_t),H_t=0)
    \end{equation*}
    \begin{equation*}
        \widehat{\rev}_t^-(p,q)= \frac{1-(q_t-p_t)\I(s_t\le \hat{p}_t, b_t\ge \hat q_t)}{(1-\alpha)\hat{\pi}_t(p,q)}\I((p,q)= (\hat p_t,\hat q_t),H_t=0).
    \end{equation*}
    \begin{equation*}
        \hat{\ell}_t(p,q)= \hat L_t^-(p,q)+\hat R_t^-(p,q)+(\lambda_t+1)\widehat{\rev}_t^-(p,q)
    \end{equation*}
    and finally 
    \begin{equation*}
        \tilde{\ell}_t(p,q)= \tilde L_t^-(p,q)+\tilde R_t^-(p,q)+(\lambda_t+1)\widetilde{\rev}_t^-(p,q).
    \end{equation*}

First, we establish that the estimator $\hat{\ell}_t(p,q)$ is unbiased with respect to our Lagrangian objective $\ell_t(p,q)$.

\begin{lemma}[Unbiasedness of $\hat{\ell}_t$]
\label{lemma: ell hat}
Let $\hat{\ell}_t(p,q)$ be defined as:
\begin{align*}
    \hat{\ell}_t(p,q) = \hat L_t(p,q)^- + \hat R_t(p,q)^- + (\lambda_t+1)\widehat{\rev}_t(p,q)^-.
\end{align*}
Then, $\mathbb{E}[\hat{\ell}_t(p,q) ] = \ell_t(p,q)$.
\end{lemma}
\begin{proof}
    By the law of total expectation, $\mathbb{E}[\hat{\ell}_t(p,q)] = \sum_{k=0}^2 \mathbb{P}(H_t=k) \mathbb{E}[\hat{\ell}_t(p,q) \mid H_t=k]$. We analyze the components:

\textbf{Case $H_t=0$ (Bandit Component):}
\begin{align*}
    \mathbb{P}(H_t=0) \mathbb{E}[\hat{\ell}_t \mid H_t=0] &= (1-\alpha) \mathbb{E} \left[ \frac{(\lambda_t+1)(1- \rev_t(\hat{p}_t, \hat{q}_t))}{(1-\alpha)\hat{\pi}_t(p,q)} \mathbb{I}((p,q)=(\hat{p}_t,\hat{q}_t))  \right] \\
    &= (\lambda_t+1) (1-\mathbb{E}[\rev_t(p,q)] ).
\end{align*}

\textbf{Case $H_t=1, 2$ (Marginal Components):}
For $H_t=1$ (exploring the seller side while fixing the buyer's price $\hat{q}_t$):
\begin{align*}
    \mathbb{P}(H_t=1) \mathbb{E}[\hat{\ell}_t \mid H_t=1] &= \frac{\alpha}{2} \mathbb{E} \left[ \frac{1-\mathbb{I}(s_t \le U_t \le p, b_t \ge q)}{\frac{\alpha}{2} \sum_{p'} \hat{\pi}_t(p',q)} \mathbb{I}(q=\hat{q}_t)  \right] \\
    &= (1-\mathbb{E}[L_t(p,q)]) .
\end{align*}
Similarly, for $H_t=2$ (exploring the buyer side):
\[ \mathbb{P}(H_t=2) \mathbb{E}[\hat{\ell}_t \mid H_t=2] = (1-\mathbb{E}[R_t(p,q)]) . \]
Summing these
\begin{align*}
    \mathbb{E}[\hat{\ell}_t(p,q)] &= (\lambda_t+1)(1- \rev_t(p,q)) +(1- L_t(p,q)) + (1-R_t(p,q))  = \ell_t(p,q).
\end{align*}
\end{proof}

Let us define these new quantities :
\[\tilde{L}^-_t(q)=\frac{1-\I(s_t\le U_t, b_t\ge q)}{\frac{\alpha}{2}\sum_{p':(p,q)\in \cA}\hat{\pi}_t(p',q)+\gamma}\I(q=\hat q_t,H_t=1)\]
so that $ \tilde{L}^-_t(p,q)\le \tilde{L}^-_t(q),$
 and similarly,
\[\tilde{R}^-_t(p)=\frac{1-\I(s_t\le p, b_t\ge V_t)}{\frac{\alpha}{2}\sum_{q':(p,q')\in \cA}\hat{\pi}_t(p,q')+\gamma}\I(p=\hat p_t,H_t=2)\]
so that $ \tilde{R}^-_t(p,q)\le \tilde{R}^-_t(p).$

This means that we can use them to give the following result.
\begin{lemma}
\label{lemma: aux hat ell- tilde ell}
    \begin{align*}
         \sum_{(p,q)\in \mathcal{A}}& \hat\pi_t(p,q) \left( \hat{\ell}_t(p,q)-\tilde{\ell}_t(p,q) \right) \le \gamma \left(\frac{2}{\alpha}\sum_{q\in \mathcal{A}_b}\widetilde{L}^-_t(q)+\frac{2}{\alpha}\sum_{q\in \mathcal{A}_s}\widetilde{R}^-_t(p)+\frac{(1+\lambda_t)}{1-\alpha}\hspace{-0.2cm}\sum_{(p,q)\in \mathcal{A}}\widetilde{\rev}^-_t(p,q)\right)
    \end{align*}
\end{lemma}
\begin{proof}
We prove the desired inequality by focusing at each component of $\hat \ell$ and $\tilde{\ell}$.

Therefore,
        \begin{align*}
      &   \sum_{(p,q)\in \mathcal{A}} \hat\pi_t(p,q) \left( \hat{L}^-_t(p,q)-\tilde{L}^-_t(p,q) \right) \\
         &\le \sum_{(p,q)\in \mathcal{A}} \hat\pi_t(p,q) \cdot\\
         & \hspace{1cm}\cdot\left( \frac{1-\I(s_t\le U_t\le p, b_t\ge q)}{\frac{\alpha}{2}\sum_{p':(p,q)\in \cA}\hat{\pi}_t(p',q)}\I(q=\hat q_t,H_t=1)-\frac{1-\I(s_t\le U_t\le p, b_t\ge q)}{\frac{\alpha}{2}\sum_{p':(p,q)\in \cA}\hat{\pi}_t(p',q)+\gamma}\I(q=\hat q_t,H_t=1) \right)\\
         & \le  \sum_{(p,q)\in \mathcal{A}} \hat\pi_t(p,q)\left(1-\I(s_t\le U_t\le p, b_t\ge q)\right)\I(q=\hat q_t,H_t=1)\cdot\\
         & \hspace{6.7cm}\cdot\underbrace{\left(\frac{1}{\frac{\alpha}{2}\sum_{p':(p,q)\in \cA}\hat{\pi}_t(p',q)}-\frac{1}{\frac{\alpha}{2}\sum_{p':(p,q)\in \cA}\hat{\pi}_t(p',q)+\gamma}\right)}_{(\star)}
    \end{align*}
and $(\star)$ can be further bounded as 
    \begin{align*}
        (\star)=&\left(\frac{1}{\frac{\alpha}{2}\sum_{p':(p,q)\in \cA}\hat{\pi}_t(p',q)}-\frac{1}{\frac{\alpha}{2}\sum_{p':(p,q)\in \cA}\hat{\pi}_t(p',q)+\gamma}\right)\\
        & = \frac{1}{\frac{\alpha}{2}\sum_{p':(p,q)\in \cA}\hat{\pi}_t(p',q)}\left(\frac{\frac{\alpha}{2}\sum_{p':(p,q)\in \cA}\hat{\pi}_t(p',q)+\gamma-\frac{\alpha}{2}\sum_{p':(p,q)\in \cA}\hat{\pi}_t(p',q)}{\frac{\alpha}{2}\sum_{p':(p,q)\in \cA}\hat{\pi}_t(p',q)+\gamma}\right)\\
        & \le \frac{1}{\frac{\alpha}{2}\sum_{p':(p,q)\in \cA}\hat{\pi}_t(p',q)}\left(\frac{\gamma}{\frac{\alpha}{2}\sum_{p':(p,q)\in \cA}\hat{\pi}_t(p',q)+\gamma}\right)\\
    \end{align*}

    Therefore 
    \begin{align*}
        &   \sum_{(p,q)\in \mathcal{A}} \hat\pi_t(p,q) \left( \hat{L}^-_t(p,q)-\tilde{L}^-_t(p,q) \right) \\
        & \le \gamma \sum_{(p,q)\in \mathcal{A}} \frac{\hat{\pi}_t(p,q)}{\frac{\alpha}{2}\sum_{p':(p,q)\in \cA}\hat{\pi}_t(p',q)}\tilde{L}_t^-(p,q)\le \gamma \frac{2}{\alpha}\sum_{q \in \cA_b}\tilde{L}_t^-(q).
    \end{align*}
    With analogous analysis
    \begin{align*}
          \sum_{(p,q)\in \mathcal{A}} \hat\pi_t(p,q) \left( \hat{R}^-_t(p,q)-\tilde{R}^-_t(p,q) \right) \le \gamma \frac{2}{\alpha}\sum_{p \in \cA_s}\tilde{R}_t^-(p),
    \end{align*}
    \begin{align*}
        \sum_{(p,q)\in \mathcal{A}} \hat\pi_t(p,q) \left( \widehat{\rev}^-_t(p,q)-\widetilde{\rev}^-_t(p,q) \right) \le \gamma \frac{1}{1-\alpha}\sum_{(p,q) \in \cA}\widetilde{\rev}^-_t(p,q).
    \end{align*}
    Finally, we recalling the  definition of $\hat\ell_t(p,q)$ and $\tilde\ell_t(p,q)$ we conclude the proof.
\end{proof}

Then, we can use Lemma 1 from \cite{neu2015explore}, on each component, obtaining the following result.

\begin{lemma}\label{lemma: neu 2}
    Let  $ \gamma \geq 0$  and let $ \alpha_{t,i} $ be nonnegative random variables satisfying $ \alpha^1_{t}(p,q) \leq 2\gamma,\;\alpha^2_{t}(q) \leq 2\gamma,\;\alpha^3_{t}(p) \leq 2\gamma$  for all $ t $ and  $(p,q)\in \cA$. Then, with probability at least $1-1/T$, for all $t\in [T]$
\begin{equation*}
\sum_{t=1}^{\tau} \sum_{q \in \cA_b} \alpha^1_{t}(q) (\tilde{L}^-_t(q) - {L}^-_t(q)) \leq \log(3T^2),
\end{equation*}
\begin{equation*}
\sum_{t=1}^{\tau} \sum_{p\in \cA_s} \alpha_{t}^2(p) (\tilde{R}^-_t(p) - {R}^-_t(p)) \leq \log(3T^2),
\end{equation*}
\begin{equation*}
\sum_{t=1}^{\tau} \sum_{(p,q)\in \cA} \alpha_{t}^3(p,q) (\widetilde{\rev}^-_t(p,q) - {\rev}^-_t(p,q)) \leq 2\log(3T^2).
\end{equation*}
\end{lemma}
\begin{proof}
    By definition $L_t(p,q)^- , R_t(p,q)^- ,\frac{1}{2}{\rev}_t(p,q)^-$ are all in $[0,1]$, therefore we can apply Lemma 1 of \cite{neu2015explore}, with a union bound over the rounds $\tau\in [T]$, and over all three type of events.
\end{proof}
In addition we give the following result.
\begin{lemma} \label{lemma: neu 1}
    Let  $ \gamma \geq 0$  and let $ \alpha_{t,i} $ be nonnegative random variables satisfying $ \alpha_{t}(p,q) \leq 2\gamma$  for all $ t $ and  $(p,q)\in \cA$. Then, with probability at least $1-1/T$, for all $t\in [T]$
\begin{equation*}
\sum_{t=1}^{\tau} \sum_{(p,q)\in \cA} \alpha_{t}(p,q) (\tilde \ell_t(p,q) - \ell_t(p,q)) \leq (4+2M)\log(3T/\delta)
\end{equation*}
where $M=\max_{t\in [T]}|\lambda_t|$.
\end{lemma}
\begin{proof}

The proof of the lemma follows directly by applying \Cref{lemma: neu 2} to the definition of $\tilde{\ell}_t$ and $\ell_t$.
In particular since $\ell_t(p,q)\le 2+2(1+\lambda_t)$ deterministically.
\end{proof}

In addition, similar to \cite{neu2015explore} we have the following Corollary.

\begin{corollary}\label{corollary: neu}
    For all $(p,q)\in \cA$ it holds simultaneously with probability at least $1-1/T$:
\begin{align*}
    \sum_{t=1}^\tau\left(\widetilde{\ell}_t(p,q)-{\ell}_t(p,q)\right)\le \frac{(4M+2)\ln\left(3|\cA|T^2\right)}{2\gamma},
\end{align*}
where $M\ge \max_{t\in [T]}|\lambda_t|$.
\end{corollary}

\subsection{Bounding the Second Moment of the Estimators}
The regret of EXP3.IX depends on the second moment of the estimators $\{\tilde\ell_t\}_t$, so in this section we will do exactly that.

Let's recall the definition of  $\tilde{L}^-_t(q),\tilde{R}^-_t(p)$ as 
\[\tilde{L}^-_t(q)=\frac{1-\I(s_t\le U_t, b_t\ge q)}{\frac{\alpha}{2}\sum_{p':(p,q)\in \cA}\hat{\pi}_t(p',q)+\gamma}\I(q=\hat q_t,H_t=1)\]
and
\[\tilde{R}^-_t(p)=\frac{1-\I(s_t\le p, b_t\ge V_t)}{\frac{\alpha}{2}\sum_{q':(p,q')\in \cA}\hat{\pi}_t(p,q')+\gamma}\I(p=\hat p_t,H_t=2).\]

\begin{lemma} \label{lemma: bound l tilde 2}
For all $t\in [T]$, the second moment of the biased estimators satisfies:
\begin{align*}
     \sum_{(p,q)\in \mathcal{A}} \hat{\pi}_t(p,q) \left( \tilde{\ell}_t(p,q) \right)^2 
    &\le 
     \left( \frac{6}{\alpha} \sum_{q \in \mathcal{A}_b} \tilde{L}_t^-(q) + \frac{6}{\alpha} \sum_{p \in \mathcal{A}_s} \tilde{R}_t^-(p) + \frac{3(\lambda_t+1)^2}{1-\alpha} \sum_{(p,q) \in \mathcal{A}} \widetilde{\rev}_t^-(p,q) \right).
\end{align*}

\end{lemma}

\begin{proof}
Applying the Cauchy-Schwarz inequality, $(a+b+c+d)^2 \le 4(a^2 + b^2 + c^2 + d^2)$, we bound the squared terms of the estimator. 

\begin{align*}
    \sum_{(p,q)\in \cA}&\hat\pi_t(p,q) \left(\tilde{\ell}_t(p,q)\right)^2 \hspace{-0.2cm}\le 3\hspace{-0.2cm}\sum_{q\in \cA_b}\sum_{p: (p,q)\in \mathcal{A}}\hspace{-0.2cm}\hat\pi_t(p,q)\left(\frac{1-\I(s_t\le U_t\le p, b_t\ge q)}{\frac{\alpha}{2}\sum_{p':(p,q)\in \cA}\hat{\pi}_t(p',q)+\gamma}\I(q=\hat q_t,H_t=1)\right)^2 \\
    & + 3\sum_{p\in \cA_s}\sum_{q:(p,q)\in \cA}\hat\pi_t(p,q)\left(\frac{1-\I(s_t\le p, b_t\ge V_t\ge q)}{\frac{\alpha}{2}\sum_{q': (p,q')\in \mathcal{A}}\hat{\pi}_t(p,q')+\gamma}\I(p= \hat p_t,H_t=2)\right)^2\\
    & + 3(\lambda_t+1)^2 \hspace{-0.2cm}\sum_{(p,q)\in \cA}\hat\pi_t(p,q)\left(\frac{1-(q_t-p_t)\I(s_t\le \hat{p}_t, b_t\ge \hat q_t)}{(1-\alpha)\hat{\pi}_t(p,q)+\gamma}\I((p,q)= (\hat p_t,\hat q_t),H_t=0)\right)^2\\
    & \le \frac{6}{\alpha}\sum_{q\in \cA_b}\widetilde{L}^-_t(q)+\frac{6}{\alpha}\sum_{p\in \cA_s}\widetilde{R}^-_t(p)+ \frac{3(\lambda_t+1)^2}{1-\alpha} \sum_{(p,q)\in \cA}\widetilde{\rev}^-_t(p,q).
\end{align*}
\end{proof}

\subsection{Proof of \Cref{theo: RegretPrimal}}
We are ready to prove \Cref{theo: RegretPrimal}.

\RegretPrimal*
\begin{proof}
In this proof, we assume $\alpha\le 1/2$, so that given any $C\in \R$ , $C/(1-\alpha)\le 2C= \cO(C)$. Because of this we will omit the dependence on $\frac{1}{1-\alpha}$
First, we observe that the regret with respect to the lagrangian function is equivalent to the regret with respect to the $\ell_t$ loss functions.
\begin{align*}
        \sum_{t=1}^{\tau}&\sum_{(p,q)\in \cA}\pi(p,q)\left(\gft_t(p,q)+\lambda_t\rev_t(p,q)\right)-\sum_{t=1}^{\tau}\sum_{(p,q)\in \cA}\left(\gft_t(p_t,q_t)+\lambda_t\rev_t(p_t,q_t)\right)\\
        & = \sum_{t=1}^{\tau}\sum_{(p,q)\in \cA}\pi(p,q)\left(L_t(p,q)+R_t(p,q)+(\lambda_t+1)\rev_t(p,q)\right)\\
        & \quad-\sum_{t=1}^{\tau}\left(L_t(p_t,q_t)+R_t(p_t,q_t)+(\lambda_t+1)\rev_t(p_t,q_t)\right)\\
        & = -\sum_{t=1}^{\tau}\sum_{(p,q)\in \cA}\pi(p,q)\left((1-L_t(p,q))+(1-R_t(p,q))+(\lambda_t+1)(1-\rev_t(p,q))\right) \\
        &\quad+ \sum_{t=1}^{\tau}\sum_{(p,q)\in \cA}\pi(p,q)(3+\lambda_t)\\
        & \quad+\sum_{t=1}^{\tau}\left((1-L_t(p_t,q_t))+(1-R_t(p_t,q_t))+(\lambda_t+1)(1-\rev_t(p_t,q_t))\right) - \sum_{t=1}^{\tau}(3+\lambda_t)\\
        & = \sum_{t=1}^{\tau}\left(\ell_t(p_t,q_t)-\sum_{(p,q)\in \cA}\pi(p,q)\ell_t(p,q)\right).
    \end{align*}
Hence, we decompose the regret with respect to the loss functions $\ell_t$ in manageable parts. 

For all $(p^\dagger,q^\dagger)\in \cA$, for all $\tau\in [T]$
\begin{align*}
    \sum_{t=1}^\tau \ell_t(p_t,q_t)&- \ell_t(p^\dagger,q^\dagger)\le \underbrace{\sum_{t=1}^\tau \left(\ell_t(p_t,q_t) - \sum_{(p,q)\in \cA}{\pi}_t(p,q)\ell_t(p,q)\right)}_{(1)}\\
    &+ \underbrace{\sum_{t=1}^\tau \left(\sum_{(p,q)\in \cA}{\pi}_t(p,q)\ell_t(p,q) - \sum_{(p,q)\in \cA}\hat{\pi}_t(p,q)\ell_t(p,q)\right)}_{(2)} \\
    & + \underbrace{\sum_{t=1}^\tau \left(\sum_{(p,q)\in \cA}\hat{\pi}_t(p,q)\ell_t(p,q)-\sum_{(p,q)\in \cA}\hat{\pi}_t(p,q)\tilde \ell_t(p,q)\right)}_{(3)}\\
    & + \underbrace{\sum_{t=1}^\tau \left(\sum_{(p,q)\in \cA}\hat{\pi}_t(p,q)\tilde\ell_t(p,q)-\tilde \ell_t(p^\dagger,q^\dagger)\right)}_{(4)}\\
    & + \underbrace{\sum_{t=1}^\tau \left(\tilde \ell_t(p^\dagger,q^\dagger)-\ell_t(p^\dagger,q^\dagger)\right)}_{(5)}.
\end{align*}

\paragraph{First Term}
The first term can be bounded by applying Azuma-Hoeffding inequality, and with probability at least $1-\delta'$
\begin{align*}
    (1)\le  \tilde\cO(M\sqrt{T}).
\end{align*}

\paragraph{Second Term} The second term can be bounded by definition of ${\pi}_t$ and $\hat{\pi}_t$. Indeed,
\begin{align*}
    (2) &= \sum_{t=1}^\tau \left(\sum_{(p,q)\in \cA}{\pi}_t(p,q)\ell_t(p,q) - \sum_{(p,q)\in \cA}\hat{\pi}_t(p,q)\ell_t(p,q)\right) \\
    & \le \sum_{t=1}^\tau \left(\sum_{(p,q)\in \cA}(1-\alpha)\hat{\pi}_t(p,q)\ell_t(p,q) - \sum_{(p,q)\in \cA}\hat{\pi}_t(p,q)\ell_t(p,q)\right)+\tau \cdot \alpha \le \tau \cdot \alpha,
\end{align*}
since $\ell_T(p,q)\ge 0$ for all $t\in [T],(p,q)\in \cA$.

\paragraph{Third term} To bound the third term we can divide this term in two, employing the unbiased loss estimators $\{\hat{\ell}_t\}$.
\begin{align*}
    (3)&= \sum_{t=1}^\tau \left(\sum_{(p,q)\in \cA}\hat{\pi}_t(p,q)\ell_t(p,q)-\sum_{(p,q)\in \cA}\hat{\pi}_t(p,q)\hat \ell_t(p,q)\right)\\
    & \quad+\sum_{t=1}^\tau \left(\sum_{(p,q)\in \cA}\hat{\pi}_t(p,q)\hat\ell_t(p,q)-\sum_{(p,q)\in \cA}\hat{\pi}_t(p,q)\tilde \ell_t(p,q)\right),
\end{align*}
of which the second term is bounded by \Cref{lemma: aux hat ell- tilde ell} with probability at least $1-1/T$ as
\begin{align*}
    \sum_{t=1}^\tau &\left(\sum_{(p,q)\in \cA}\hat{\pi}_t(p,q)\hat\ell_t(p,q)-\sum_{(p,q)\in \cA}\hat{\pi}_t(p,q)\tilde \ell_t(p,q)\right)\\
    & \quad \le \sum_{t=1}^\tau\gamma \left(\frac{2}{\alpha}\sum_{q\in \mathcal{A}_b}\widetilde{L}^-_t(q)+\frac{2}{\alpha}\sum_{q\in \mathcal{A}_s}\widetilde{R}^-_t(p)+\frac{(1+\lambda_t)}{1-\alpha}\sum_{(p,q)\in \mathcal{A}}\widetilde{\rev}^-_t(p,q)\right)\\
    & \quad\le \sum_{t=1}^\tau\gamma \left(\frac{2}{\alpha}\sum_{q\in \mathcal{A}_b}{L}^-_t(q)+\frac{2}{\alpha}\sum_{q\in \mathcal{A}_s}{R}^-_t(p)+\frac{(1+\lambda_t)}{1-\alpha}\sum_{(p,q)\in \mathcal{A}}{\rev}^-_t(p,q)\right)\\
    & \quad\quad+ (4+2M)\gamma \log(3T^2)\\
    & \quad\le \tau \gamma \left(\frac{2(|\cA_b|+|\cA_s|)}{\alpha}+\frac{(M+1)}{1-\alpha}|\cA|\right) + (4+2M)\gamma \log(3T^2),
\end{align*}
where the second inequality hold with probability at least $1-1/T$ by \Cref{lemma: neu 2}.
Then we can apply Azuma-Hoeffding inequality to the first term and with probability at least $1-1/T$
\begin{align*}
    \sum_{t=1}^\tau \left(\sum_{(p,q)\in \cA}\hat{\pi}_t(p,q)\ell_t(p,q)-\sum_{(p,q)\in \cA}\hat{\pi}_t(p,q)\hat \ell_t(p,q)\right)\le \tilde\cO\left(\frac{M}{\alpha}\sqrt{T}\right),
\end{align*}
since $\E\left[\sum_{(p,q)\in \cA}\hat{\pi}_t(p,q)\hat\ell_t(p,q)\right]= \sum_{(p,q)\in \cA}\hat{\pi}_t(p,q) \ell_t(p,q)$ by \Cref{lemma: ell hat}, and that 
\begin{align*}
         0\le \sum_{(p,q)\in \mathcal{A}}& \hat\pi_t(p,q) \left( \hat{\ell}_t(p,q) \right) \\
         &\le \sum_{(p,q)\in \mathcal{A}} \hat\pi_t(p,q) \left( \hat{L}^-_t(p,q)+\hat{R}^-_t(p,q)+(1+\lambda_t)\widehat{\rev}^-_t(p,q) \right)\\
         &\le \sum_{(p,q)\in \mathcal{A}} \hat\pi_t(p,q) \left( \hat{L}^-_t(q)+\hat{R}^-_t(p)+(1+\lambda_t)\widehat{\rev}^-_t(p,q) \right)\\
         & = \sum_{q':(\hat{p}_t,q)\in \mathcal{A}} \hat\pi_t(\hat{p},q')\hat{L}^-_t(\hat q_t) +  \sum_{p':(p',\hat{p}_t)\in \mathcal{A}} \hat\pi_t(p',\hat{q}_t)\hat{L}^-_t(\hat q_t)+ \pi_t(\hat{p}_t,\hat{q}_t)\widehat{\rev}^-_t(\hat p_t,\hat q_t) \\
         & \le \frac{2}{\alpha}+2\frac{(\lambda_t+1)}{1-\alpha} \le \frac{2}{\alpha}+ \frac{2M+2}{1-\alpha}.
    \end{align*} 

\paragraph{Fourth term} The fourth term is regret term generated by EXP3. 
\begin{align*}\label{eq: exp3}
        (4)&=\sum_{t=1}^\tau \left(\sum_{(p,q)\in \cA}\hat{\pi}_t(p,q)\tilde{\ell}_t(p,q)-\tilde \ell_t(p^\dagger,q^\dagger)\right)\le\frac{\log |\cA|}{\eta^P}+\sum_{t=1}^\tau \frac{\eta^P}{2}\sum_{(p,q)\in \cA}\hat{\pi}_t(p,q)\tilde \ell_t(p,q)^2\\
        & \le \frac{\log |\cA|}{\eta^P}+\sum_{t=1}^\tau \frac{\eta^P}{2}\left( \frac{6}{\alpha} \sum_{q \in \mathcal{A}_b} \tilde{L}_t^-(q) + \frac{6}{\alpha} \sum_{p \in \mathcal{A}_s} \tilde{R}_t^-(p) + \frac{3(\lambda_t+1)^2}{1-\alpha} \sum_{a \in \mathcal{A}} \widetilde{\rev}_t^-(a) \right)\\
        & \le \frac{\log |\cA|}{\eta^P}+\sum_{t=1}^\tau \frac{\eta^P}{2}\left( \frac{6}{\alpha} \sum_{q \in \mathcal{A}_b} {L}_t^-(q) + \frac{6}{\alpha} \sum_{p \in \mathcal{A}_s} {R}_t^-(p) + \frac{3(\lambda_t+1)^2}{1-\alpha} \sum_{a \in \mathcal{A}} {\rev}_t^-(a) \right)\\
        & +\left(\frac{3}{\alpha}+\frac{3}{\alpha}+\frac{6(M+1)^2}{1-\alpha}\right)\log(3T^2)\\
        & \le \frac{\log |\cA|}{\eta^P}+ \frac{\tau \cdot \eta^P}{2}\left( \frac{6|\cA_b|}{\alpha}  + \frac{6|\cA_s|}{\alpha} + \frac{3(M+1)^2}{1-\alpha} |\cA| \right)\\
        &+\left(\frac{6}{\alpha}+\frac{6(M+1)^2}{1-\alpha}\right)\log(3T^2),
\end{align*}
    where the first inequality holds by \Cref{lemma: bound l tilde 2}:
    \begin{align*}
        \sum_{(p,q)\in \mathcal{A}} \hat{\pi}_t(p,q) \left( \tilde{\ell}_t(p,q) \right)^2 
    &\le 
     \left( \frac{6}{\alpha} \sum_{q \in \mathcal{A}_b} \tilde{L}_t^-(q) + \frac{6}{\alpha} \sum_{p \in \mathcal{A}_s} \tilde{R}_t^-(p) + \frac{3(\lambda_t+1)^2}{1-\alpha} \sum_{a \in \mathcal{A}} \widetilde{\rev}_t^-(a) \right).
    \end{align*}
    and the second inequality holds with probability $1-1/T$ by \Cref{lemma: neu 2}.

\paragraph{Fifth term} The fifth and last component is a direct application of  \Cref{corollary: neu}. With probability at least $1-1/T$
\begin{align*}
    (5)\le \frac{(4M+2)\ln\left(3|\cA|T^2\right)}{2\gamma}.
\end{align*}

\paragraph{Conclusion } Finally, with $\eta^P=2\gamma= \frac{1}{M}\sqrt{\frac{log(|\cA|)}{|\cA|T}}$, by union bound over the events, we have with probability at least $1-\cO(1/T)$
for all $\pi\in \Delta(\cA)$, for all $t\in [T]$
\begin{align*}
    \sum_{t=1}^\tau \left(\ell_t(p_t,q_t)- \sum_{(p,q)\in \cA}\pi(p,q)\ell_t(p,q)\right)\le 
\tilde{\cO}\left(M\left(\sqrt{|\cA|T}+\frac{|\cA_s|+|\cA_b|}{\sqrt{|\cA|}}\frac{\sqrt{T}}{\alpha}+\alpha \cdot T\right)\right).
\end{align*}

\end{proof}

\section{Proof of \Cref{theo: main}}
To prove \Cref{theo: main} we will decompose the regret of \Cref{alg:main} in its components, and then we will show that there exist a choice of parameters such that \Cref{theo: main} is verified.

\begin{restatable}{lemma}{regretdecomp}
With probability at least $1-\cO(1/T)$, the regret of Algorithm \ref{alg:main} satisfies:
\begin{align*}
    \opt^{\,\textnormal{D}}&-\sum_{t=1}^T \gft_t(p_t,q_t) \le R_T^P + 2R_T^D + 2M\cdot C + \tilde \cO(T^{3/4}) + R_K\\
\end{align*}
and 
\begin{align*}
    \opt^{\textnormal{F}}&-\sum_{t=1}^T \gft_t(p_t,q_t)\le  R_T^P + 2R_T^D +\tilde{\cO}(T^{3/4}) + R_{0,K}.\\
\end{align*}
where $M =16\log(T)$ is the bound on the dual variables, $R_T^P$ and $R_T^D$ are the primal and dual regrets respectively, $R_K$ and $R_{0,K}$ are the regret generated by the approximation.
\end{restatable}

\begin{proof}

We consider two possible scenarios. The first is the trivial case in which $\cT_1=[T]$ and $\cT_2=\emptyset$, so the algorithm never exit the procedure \textsf{Rev-Max}. This case is a trivial one, as it is sufficient to observe that $\sum_{t\in [T]}\rev_t(p_t,q_t)=\sum_{t\in \cT_1}\rev_t(p_t,q_t) \le 1$. Hence, by \Cref{lemma: profit max}, in this case with probability at least $1-1/T$
\[\opt^{\textnormal{F}}\le \opt^{\,\textnormal{D}}\le \tilde\cO(T^{3/4}),\]
which implies the statement for this case, as in $\cT_1$ the algorithm goal is to maximize the revenue and post prices that never generates negative GFT (posts only $(p,q)$ such that $q\ge q$). 

We focus therefore on the other non-pathological case, in which $\cT_2\neq\emptyset$.

Let $\pi\in \Delta(\cA)$, we have with probability at least $1-\cO(1/T)$ \Cref{eq: primal 2}
\begin{align*}
    \sum_{t\in \cT_2}&\left( \mathbb{E}_{(p,q)\sim\pi}[\gft_t(p,q)] \right) - \sum_{t\in \cT_2}\left(\gft_t(p_t,q_t)\right)\\
       & \le   R_T^P  -\sum_{t\in \cT_2}\lambda_t\mathbb{E}_{(p,q)\sim\pi}[\rev_t(p,q)] + \sum_{t\in \cT_2}\lambda_t\rev_t(p_t,q_t)\\
       & \le R_T^P  -\sum_{t\in \cT_2}\lambda_t\mathbb{E}_{(p,q)\sim\pi}[\rev_t(p,q)] - M\left(\sum_{t\in \cT_1}\rev_t(p_t,q_t)\right)+ 2M + 2R_T^D,
\end{align*}
where the last inequality employs \Cref{lemma: dual in regret}.

Let's consider now the regret generated by playing the revenue maximization procedure. 
If $\pi$ is feasible (either globally or per round),
by \Cref{lemma: profit max} with probability at least $1-1/T$
\begin{align*}
    \sum_{t\in \cT_1}\mathbb{E}_{(p,q)\sim \pi}[\gft(p,q)]\le& 16 \log T\cdot \left(\sum_{t\in \cT_1}\rev_t(p_t,q_t)+1\right) + \tilde{\cO}(T^{3/4}).
\end{align*}

Therefore, joining everything, and considering $M= 16\log(T)$ it holds with probability at least $1-\delta$

\begin{align*}
    \opt^{\,\textnormal{D}}&-\sum_{t=1}^T \gft_t(p_t,q_t)\le R_T^P + 2R_T^D - M\sum_{t\in \cT_1}\rev_t(p_t,q_t)+2M\\
    & +  M\sum_{t\in \cT_1}\rev_t(p_t,q_t)+2M\cdot C + \tilde\cO(T^{3/4}) + R_K\\
    & \le \cdot R_T^P + 2R_T^D + 2M\cdot C + \tilde\cO(T^{3/4}) + R_K
\end{align*}
and 
\begin{align*}
    \opt^{\textnormal{F}}&-\sum_{t=1}^T \gft_t(p_t,q_t)\le  R_T^P + 2R_T^D + \tilde{\cO}(T^{3/4}) + R_{0,K}.
\end{align*}
\end{proof}

To conclude the proof of \Cref{theo: main}, it sufficient to recall that with probability at least $1-\cO(1/T)$
\begin{itemize}
    \item $R_T^P=\tilde\cO\left(\log(K)\left(\alpha\cdot T + K\sqrt{T}+\frac{\sqrt{T}}{\alpha}\right)\right)$
    \item $R_T^D=\tilde \cO\left(\log(K)\sqrt{T}\right)$
    \item $M=\cO(\log(T))$
    \item $R_K=\tilde{\cO}(\frac{T}{K}+C)$ if the distributions are $\sigma$-smooth
    \item $R_{0,K}=\cO(\frac{T}{K}).$
\end{itemize}

Hence, by setting $K=T^{-1/4}$ and $\alpha=T^{-1/4}$, we get the desired results.

\end{document}